\documentclass[11pt]{article}
\usepackage{inconsolata}
\usepackage{microtype}
\usepackage{libertine}
\usepackage[margin=1in]{geometry}
\usepackage[utf8]{inputenc}
\usepackage{authblk}
\usepackage{amsmath, amssymb, amsthm, thmtools, amsfonts, bm, bbm, thm-restate}
\usepackage{algorithm}
\usepackage{algorithmicx}
\usepackage[noend]{algpseudocode}
\usepackage{cite}
\usepackage[numbers,sort]{natbib}

\usepackage{asymptote}
\usepackage{graphicx}
\usepackage{todonotes}
\usepackage{multirow}
\usepackage{comment}
\usepackage{dsfont}
\usepackage{bigstrut}
\usepackage{caption}
\usepackage{subcaption}
\usepackage{multirow}
\usepackage{makecell}
\usepackage{booktabs}
\usepackage{xcolor}
\definecolor{BrickRed}{rgb}{0.8,0.25,0.33}
\usepackage[pagebackref,colorlinks,citecolor=blue,linkcolor=BrickRed,hypertexnames=true,plainpages=false]{hyperref}

\usepackage{enumitem}
\usepackage[normalem]{ulem}

\makeatletter
\providecommand\theHALG@line{\thealgorithm.\arabic{ALG@line}}
\makeatother

\usepackage[framemethod=tikz]{mdframed}
\usepackage{nicefrac}

\usepackage{tikz}
\usetikzlibrary{positioning, fit, calc}

\usepackage{pifont}

\usepackage[capitalize]{cleveref}

\declaretheorem[numberwithin=section,refname={Theorem,Theorems},Refname={Theorem,Theorems}]{theorem}
\declaretheorem[numberwithin=section,refname={Theorem,Theorems},Refname={Theorem,Theorems}]{thm}
\declaretheorem[numberlike=theorem]{lemma}

\declaretheorem[numberlike=theorem]{corollary}
\theoremstyle{definition}
\declaretheorem[numberlike=theorem]{definition}
\theoremstyle{plain}

\declaretheorem[numberlike=theorem]{remark}
\declaretheorem[numberlike=theorem,refname={Fact,Facts},Refname={Fact,Facts},name={Fact}]{fact}

\declaretheorem[numberlike=theorem,refname={Result,Results},Refname={Result,Results},name={Result}]{result}

\allowdisplaybreaks

\newcommand{\E}{\mathbb{E}}

\newcommand{\eps}{\varepsilon}

\newcommand{\poly}{\mathrm{poly}}

\newcommand{\udeg}{\mathrm{badness}}
\newcommand{\C}{\mathcal{C}}
\newcommand{\starcond}{{\color{blue}(\star)}}
\newcommand{\baddeg}{\mathrm{baddeg}}

\newenvironment{wrapper}[1]
{
	\smallskip
	\begin{center}
		\begin{minipage}{\linewidth}
			\begin{mdframed}[hidealllines=true, backgroundcolor=gray!20, leftmargin=0cm,innerleftmargin=0.35cm,innerrightmargin=0.35cm,innertopmargin=0.375cm,innerbottommargin=0.375cm,roundcorner=10pt]
				#1}
			{\end{mdframed}
		\end{minipage}
	\end{center}
	\smallskip
}

\usepackage{xspace}

\newcommand{\bad}{\emph{bad}\xspace}

\newcommand{\Calg}{\C_{\mathrm{alg}}}
\newcommand{\Cgreedy}{\C_{\mathrm{greedy}}}
\renewcommand{\ln}{\log}

\title{Online Edge Coloring: Sharp Thresholds}
\date{}
\author[1]{Joakim Blikstad}
\author[2]{Ola Svensson\thanks{Supported by the Swiss State Secretariat for Education, Research and Innovation (SERI) under contract number MB22.00054.}}
\author[2]{Radu Vintan\protect\footnotemark[\value{footnote}]}
\author[3]{David Wajc\thanks{Supported by a Taub Family Foundation ``Leader in Science and Technology'' fellowship and ISF grant 3200/24.}}

\affil[1]{CWI Amsterdam}
\affil[2]{EPFL}
\affil[3]{Technion --- Israel Institute of Technology}

\begin{document}
\maketitle

\pagenumbering{gobble}

\begin{abstract}
Vizing’s theorem guarantees that every graph with maximum degree~$\Delta$ admits an edge coloring using~$\Delta + 1$ colors. In online settings---where edges arrive one at a time and must be colored immediately---a simple greedy algorithm uses at most~$2\Delta - 1$ colors. Over thirty years ago, Bar-Noy, Motwani, and Naor~[IPL'92] proved that this guarantee is optimal among deterministic algorithms when~$\Delta = O(\log n)$, and among randomized algorithms when~$\Delta = O(\sqrt{\log n})$. While deterministic improvements seemed out of reach, they conjectured that for graphs with~$\Delta = \omega(\log n)$, randomized algorithms can achieve~$(1 + o(1))\Delta$ edge coloring.
This conjecture was recently resolved in the affirmative: a~$(1 + o(1))\Delta$-coloring is achievable online using randomization  for all graphs with~$\Delta = \omega(\log n)$~[BSVW STOC'24]. 

\smallskip 

Our results go further, uncovering two findings not predicted by the original conjecture. First, we give a \emph{deterministic} online algorithm achieving~$(1 + o(1))\Delta$-colorings for all~$\Delta = \omega(\log n)$. Second, we give a \emph{randomized} algorithm achieving~$(1 + o(1))\Delta$-colorings already when~$\Delta = \omega(\sqrt{\log n})$.
Our results establish sharp thresholds for when greedy can be surpassed, and near-optimal guarantees can be achieved --- matching the impossibility results of [BNMN IPL'92], both deterministically and randomly. 
\end{abstract}

\newpage 
\tableofcontents
\newpage

\pagenumbering{arabic}

\section{Introduction}

In the online edge coloring problem, an unknown $n$-node graph $G = (V, E)$ of maximum degree~$\Delta$ is revealed over time, edge by edge. 
Upon the arrival of an edge, an algorithm must irrevocably assign it a color distinct from those used on adjacent edges. 
The goal is to minimize the total number of colors used, ideally approaching the offline optimum of $\Delta$ or $\Delta + 1$ colors guaranteed by Vizing’s theorem~\cite{vizing1964estimate}.\footnote{In the offline setting, a $\Delta+1$-coloring can be computed in near-linear time, by a recent breakthrough~\cite{assadi2025vizing}. In contrast, it is NP-hard to decide whether a graph can be edge-colored with $\Delta$ colors, which are necessary by pigeonhole principle \cite{holyer1981np}.} 
The greedy online algorithm, which assigns each edge an arbitrary available color or introduces a new one if needed, provides a simple upper bound: it uses at most~$2\Delta - 1$ colors—roughly twice as many as the offline optimum.

Somewhat discouragingly, Bar-Noy, Motwani, and Naor~\cite{bar1992greedy} showed that no online algorithm can improve upon greedy's bound by even a single color. However, this lower bound holds only for graphs of low maximum degree, namely~$\Delta = O(\log n)$ for deterministic algorithms and~$\Delta = O(\sqrt{\log n})$ for randomized algorithms. 
Aware of the limitations of their lower bounds, they conjectured that better online edge coloring algorithms exist for graphs of larger degree. Specifically, while they regarded deterministic improvements over greedy as implausible, they conjectured the existence of randomized algorithms which, applied to graphs of \emph{super-logarithmic} maximum degree, output a coloring with near-optimal $(1 + o(1))\Delta$ colors.\footnote{$(1+o(1))\Delta$ colors are optimal for \emph{online} algorithms, which require at least $\Delta+\Omega(\sqrt{\Delta})$ colors for any $\Delta$~\cite{cohen2019tight}.}

The online edge coloring conjecture of~\cite{bar1992greedy} has attracted significant attention over the years. It has been studied and confirmed in several restricted models: under random-order edge arrivals~\cite{aggarwal2003switch,bahmani2012online,bhattacharya2021online,kulkarni2022online,dudeja2025randomized}, and for bipartite graphs under one-sided vertex arrivals~\cite{cohen2019tight,blikstad2023simple}. Weaker guarantees have also been obtained in general graphs under vertex arrivals~\cite{saberi2021greedy} and edge arrivals~\cite{kulkarni2022online}. The conjecture was recently resolved in~\cite{blikstad2024online}, who showed that a~$(1 + o(1))\Delta$-edge coloring is achievable online  using a \emph{randomized} algorithm for all graphs with~$\Delta = \omega(\log n)$.

Interestingly, while the conjecture was stated for randomized algorithms, it matches the known lower bound only for \emph{deterministic} algorithms, where~\cite{bar1992greedy} proved that no algorithm can improve over greedy for~$\Delta = O(\log n)$. For randomized algorithms, the best known lower bound still only holds for~$\Delta = O(\sqrt{\log n})$, leaving open the possibility of achieving $(1 + o(1))\Delta$-colorings in a broader regime.
However, we note that algorithmically matching the lower bounds of \cite{bar1992greedy} would separate online edge coloring from the more general online \emph{list} edge coloring problem, to which current online edge coloring bounds extend seamlessly \cite{blikstad2024online} (see \Cref{sec:list-edge-coloring-lower-bound}).

While~\cite{bar1992greedy} expressed skepticism about improving upon the greedy algorithm deterministically,
recent work has begun to challenge that intuition. Progress has come via reductions from randomized algorithms that succeed even against stronger, adaptive adversaries. In particular,~\cite{dudeja2025randomized} showed that the candidate randomized greedy algorithm from~\cite{bar1992greedy} achieves a~$(1 + \varepsilon)\Delta$-coloring for~$\Delta \geq Cn$ and any constants~$\varepsilon > 0$, $C \leq 1$, even when the input is generated adaptively. A classical reduction~\cite{ben1994power} then yields a deterministic algorithm with the same guarantees. For bipartite graphs under one-sided arrivals, an~$(e/(e - 1) + o(1))\Delta$-coloring is achievable for~$\Delta = \omega(\log n)$ via the same reduction~\cite{blikstad2025deterministic}.

These developments, and the absence of stronger lower bounds for randomized algorithms, leave open the possibility of resolving a stronger version of the conjecture—achieving near-optimal edge colorings deterministically, and doing so randomly already for~$\Delta = \omega(\sqrt{\log n})$. Our main results close these gaps.

\subsection{Our Results}

Our first main result shows that deterministic algorithms can achieve the same near-optimal guarantees that were previously conjectured to be achievable only via randomization:

\begin{wrapper}
\begin{result}[Deterministic Online Edge Coloring; see \Cref{thm:deterministic-precise}]\label{result:adaptive}
    There is a deterministic online $(1 + o(1))\Delta$-edge-coloring algorithm for graphs with known maximum degree~$\Delta = \omega(\log n)$.
\end{result}
\end{wrapper}

Our second main result identifies the tight threshold for randomized algorithms to achieve near-optimal edge colorings:

\begin{wrapper}
\begin{result}[Randomized Online Edge Coloring; see \Cref{thm:oblivious_theorem}]\label{result:oblivious}
    There is a randomized online $(1 + o(1))\Delta$-edge-coloring algorithm for  graphs with known maximum degree~$\Delta = \omega(\sqrt{\log n})$.
\end{result}
\end{wrapper}

Together, these results match the known lower bounds of~\cite{bar1992greedy}, who showed that for graphs with~$\Delta = O(\log n)$ and~$\Delta = O(\sqrt{\log n})$, no online algorithm—deterministic or randomized, respectively—can outperform the greedy algorithm’s palette size of $2\Delta-1$. 
The more precise theorem statements (\Cref{thm:deterministic-precise,thm:oblivious_theorem}) prove that the lower bounds of \cite{bar1992greedy} are sharp thresholds: greedy can be improved upon (1) deterministically iff $\Delta=\Omega(\log n)$ is sufficiently large, and 
    (2) randomly iff $\Delta=\Omega(\sqrt{\log n})$ is sufficiently large.
    Moreover, one can achieve $(1+o(1))\Delta$-edge-coloring iff $\Delta$ is asymptotically higher than the above thresholds.

We note that \Cref{result:oblivious} represents a departure from all previous online edge coloring algorithms that improve upon greedy, which even in special models or restricted settings required~$\Delta = \Omega(\log n)$ or higher~\cite{aggarwal2003switch,bahmani2012online,cohen2019tight,bhattacharya2021online,saberi2021greedy,kulkarni2022online,blikstad2023simple,blikstad2024online,dudeja2025randomized,blikstad2025deterministic}. Indeed, the threshold~$\Delta = \Omega(\log n)$ naturally arises in settings that rely on standard Chernoff-style concentration, giving error probabilities of~$e^{-\Delta}$, which must be union-bounded across vertices or edges. 
In contrast, our results exploit concentration over $\Theta(\Delta^2)$ many edge-color pairs, which result in error probability $\exp(-\Delta^2)$, and thus $\ll 1/\textrm{poly}(n)$ already when  $\Delta = \omega(\sqrt{\log n})$, allowing us to succeed for such smaller degrees. We expand on this intuition in the technical overview.

\subsection{Technical Overview}

Our deterministic algorithm for online edge coloring (\Cref{result:adaptive}) is obtained using the same framework recently leveraged in special cases of deterministic algorithms~\cite{dudeja2025randomized,blikstad2025deterministic}: we first design a randomized algorithm that works against an \emph{adaptive adversary}—one that chooses the next edge based on the algorithm’s previous random choices. The classical reduction of~\cite{ben1994power} then yields a deterministic algorithm with matching guarantees. The techniques we develop for this robust randomized algorithm are then adapted to also yield our tight result against \emph{oblivious adversaries} (\Cref{result:oblivious}).

To construct such a strong randomized algorithm, we diverge from most prior approaches to online edge coloring without the random-order assumption~\cite{cohen2019tight, saberi2021greedy, kulkarni2022online, blikstad2024online}, which typically reduce the problem to online matching.\footnote{The exceptions to this rule are randomized greedy \cite{dudeja2025randomized}, and some algorithms restricted to bipartite graphs under one-sided vertex arrivals \cite{blikstad2023simple,blikstad2025deterministic}.} These reductions aim to ensure that each edge is matched with probability at least $1/(\alpha\Delta)$ to obtain a coloring with roughly $\alpha \Delta$ colors. However, this matching-based framework faces inherent limitations: when faced with adaptive adversaries, even guaranteeing that an edge is matched with positive probability becomes infeasible, as the adversary can tailor future edges based on the algorithm’s randomness. Moreover, such approaches usually rely on Chernoff-type concentration bounds with failure probabilities $\exp(-\Theta(\Delta))$, which require $\Delta = \Omega(\log n)$ to union bound across vertices. These barriers prevent such reductions from being used to obtain \cref{result:adaptive,result:oblivious}.

Motivated by these challenges, we take a direct, coloring-centric approach. A natural starting point is the \emph{randomized greedy algorithm} proposed by~\cite{bar1992greedy}:
\begin{center}
\begin{minipage}{0.95\textwidth}
\begin{mdframed}[hidealllines=true, backgroundcolor=gray!15]
Fix a palette $\Calg$.\\[0.2cm]
Upon the arrival of an edge $e$, assign it a uniformly random available color from $\Calg$.
\end{mdframed}
\end{minipage}
\end{center}

The algorithm succeeds if every edge sees a non-empty set of available colors. 
While~\cite{bar1992greedy} conjectured that the algorithm succeeds with a palette size $(1+o(1))\Delta$ against oblivious adversaries when $\Delta=\omega(\log n)$, 
progress has been limited.\footnote{This conjecture was proven under random-order edge arrivals by \cite{dudeja2025randomized}, who also showed that this algorithm beats greedy for adversarially-generated graphs with $\Delta=\Omega(n)$.
Both results for this algorithm are far from the results we establish.}
Moreover, we ran several experiments that suggest a palette of size at least $\frac{e}{e - 1} \cdot \Delta$ is required for this algorithm against adaptive adversaries even for $\Delta=\omega(\log n)$. (See \cref{sec:experiments}.) Somewhat counterintuitively, uniform random choices introduce long-range bias effects, leading our experiments to indicate the need for the larger palette size. 
Similarly, in \Cref{sec:oblivious_lowerbound} we show that randomized greedy fails to even beat greedy for graphs with $\Delta=O(\log n)$ against oblivious adversaries, for which we show that $\Delta=O(\sqrt{\log n})$ is sufficient to beat greedy.

Inspired by these observations and by~\cite{blikstad2024online}, we modify the approach to use \emph{history-dependent} color probabilities, so as to ensure that each edge is assigned color $c$ with approximately the right \emph{marginal} probability $\approx 1/\Delta$. Specifically, the palette $\Calg = \{1,2, \ldots, \Delta\}$ of the algorithm is of size $\Delta$, and for each potential edge $e$ and color $c \in \Calg$, we maintain a variable $P_{ec}$ representing the intended probability of assigning color $c$ to edge $e$. These are initialized to $P_{ec} = (1-o(1))/\Delta$. They are then updated to guarantee incident edges cannot pick the same color, while maintaining the $P_{ec}$ values in expectation, by appropriate scaling, i.e.,  use of Bayes's Law.
The reason for the ``slack'' $o(1)$ will become clear when we discuss the analysis below, where we want to upper bound the probability that $\sum_{c} P_{ec} > 1$. The high-level idea of our algorithm can now be described as follows (see \Cref{alg:edge-coloring} for full details):

\begin{center}
\begin{minipage}{0.95\textwidth}
\begin{mdframed}[hidealllines=true, backgroundcolor=gray!15]
Fix a palette $\Calg = \{1,2,\dots,\Delta\}$.\\[0.2cm]
For each $e \in \binom{V}{2}$ and $c \in \Calg$, initialize $P_{ec} = (1-o(1))/\Delta$.\\[0.2cm]
When an edge $e$ arrives:
\begin{itemize}
    \item Sample a color $c$ with probability $P_{ec}$; if no color is chosen (i.e., with probability $1 - \sum_{c} P_{ec}$), color $e$ using the greedy algorithm on a backup palette $\Cgreedy$.
    \item For each adjacent edge $f$ and color $c\in \Calg$, update:
    \[
    P_{fc} = \begin{cases}
        0 & \text{if } c \text{ was assigned to } e,\\
        P_{fc} / (1 - P_{ec}) & \text{otherwise}.
    \end{cases}
    \]
\end{itemize}
\end{mdframed}
\end{minipage}
\end{center}

However, this simplified high-level description differs in two ways from our formal~\cref{alg:edge-coloring}. First, the above rule is only valid when $\sum_c P_{ec} \leq 1$. If the sum exceeds $1$, we color $e$ using the greedy algorithm on a backup palette $\Cgreedy$. To ensure that the usage of the backup palette is rare (in particular, that the number of edges colored by $\Cgreedy$ incident to each vertex is $o(\Delta)$), we analyze the martingale $Z_e = \sum_c P_{ec}$. This value $Z_e$ is intuitively the probability that $e$ is colored using $\Calg$. Its initial value is $1 - o(1)$, and for our adaptive result, we aim to use concentration bounds on the martingale $Z_e$ (see \Cref{sec:martingales} for background) to argue that $Z_e \not\in [1-o(1), 1]$ only with probability $\exp(-\Theta(\Delta))$, which is less than $1/n^3$ for $\Delta = \omega(\log n)$, allowing for a union bound over all $\binom{n}{2}$ potential edges.

The second subtlety concerns the ability to prove such concentration bounds on $Z_e$, which relies on a small \emph{step size} of the martingale $Z_e$. To enforce a small step size in~\cref{alg:edge-coloring}, we cap updates to $P_{fc}$, and do not scale up this value whenever $P_{fc} > A$ for some small threshold $A$. This prevents large jumps in $Z_e$ but introduces another complication: since some $P_{ec}$ are no longer updated, $Z_e$ develops a \emph{negative drift} and becomes a supermartingale. We call colors $c$ for which $P_{fc}>A$ \emph{bad} for $e$. The challenge then becomes to bound the number of bad colors so as to control the negative drift.

Our analysis proves that for any edge $e$, the number of bad colors is $o(\Delta)$ with high probability, even against adaptive adversaries. To achieve this, we perform a union bound over all possible neighborhoods of $e$, of which there are $n^{\Theta(\Delta)}$. For a fixed potential neighborhood of $e$, $P_{ec}$ behaves similarly to the martingales in~\cite{blikstad2024online}, and we can prove that $c$ is bad for $e$ with probability $\exp(-\Theta(\Delta))$. While this is a small probability, it is not sufficiently small to union bound over $n^{\Theta(\Delta)}$ many potential neighborhoods of $e$. 
To deal with this, we use the intuition that the events of colors becoming bad are only mildly correlated. Indeed, if these events were independent, the probability that $\epsilon \Delta$ colors become bad would be $\exp(-\Theta(\Delta))^{\epsilon \Delta} = \exp(-\Theta(\epsilon \Delta^2)) = o(1/n^{\Theta(\Delta)})$ for an appropriately chosen $\epsilon = o(1)$. While the events are not independent, carefully chosen martingales allow us to give a similar bound, and we prove that the probability that too many colors become bad simultaneously is a small enough $o(1/n^{\Delta})$ term allowing to union bound over all neighborhoods. 

\paragraph{Randomized Algorithm (Oblivious Adversary).}
Finally, we adapt these techniques to handle oblivious adversaries (\Cref{result:oblivious}). The neighborhood of each edge $e$ is now fixed in advance, eliminating the need to union bound over $n^{\Theta(\Delta)}$  many possibilities. However, we now face the issue that the probability $Z_e \not\in [1-o(1), 1]$ is $\exp(-\Theta(\Delta))$, which for $\Delta = \Theta(\sqrt{\log n})$ is $\gg 1/n^2 $, and thus will unavoidably occur for some edges.  
Ideally, we would like to prove that while this can happen, it is very unlikely to happen for more than $\epsilon \Delta$ edges incident to a vertex $v$. Indeed, if these events were independent, this would happen with probability at most $\exp\left({-\Theta(\Delta)} \right)^{\epsilon \Delta}$, which, when $\Delta = \omega(\sqrt{\log n})$ and for an appropriately chosen $\epsilon = o(1)$, is smaller than $1/n^2$ and thus allows for a union bound over all vertices. However, this time the independence intuition turns out to be misleading, because other edges incident to $v$ might affect all the $\varepsilon \Delta$ edges we are considering, thus leading to very large correlations between these edges. 

Our approach therefore needs to be different. First, we observe that the probability of a vertex $v$ becoming \emph{bad}, i.e., having $\varepsilon \Delta$ incident edges $e$ with $Z_e \notin [1-o(1), 1]$, is upper bounded by $\exp(-\Theta(\Delta))$, since this holds even for a single edge $e$ incident to $v$ turning \emph{bad}. Then, we consider an arbitrary set of $\varepsilon \Delta$ vertices. Again, if independence were to hold, this time between the considered vertices, we would have probability at most $\exp\left({-\Theta(\Delta)} \right)^{\epsilon \Delta}$ that \emph{all} of these vertices turn \emph{bad}. This time this independence intuition can be formalized by a careful and fairly complex selection of martingales. The bound on the number of \emph{bad} vertices in the neighborhood of $v$ is important because of the following. To make sure that few edges incident to a vertex are colored using the backup palette $\Calg$, the algorithm is necessarily changed so that when a vertex becomes bad, it will always color arriving incident edges if possible. This change makes edges incident to a vertex $v$ incur drifts in our analysis of the martingales and to bound this drift we need to bound the number of bad vertices (similar to us bounding the number of bad colors in the adaptive analysis).

\paragraph{Paper outline.} We provide formal definitions in~\cref{sec:prelims}. The proof of~\cref{result:adaptive} is given in~\cref{sec:adaptive}; while similar, it is cleaner and slightly simpler than the proof of~\cref{result:oblivious}, which is presented in~\cref{sec:oblivious}.

\section{Preliminaries} \label{sec:prelims}

\paragraph{Problem statement.}
In the online edge-coloring problem, we consider an unknown graph $G=(V,E)$ with $n$ vertices and maximum degree $\Delta$, with both parameters known in advance.\footnote{By \cite{cohen2019tight}, knowledge of $\Delta$ is needed to even obtain better than $1.606\Delta$ coloring.}
At each time $t=1,\dots,|E|$, an edge $e_t$ is revealed.
An online algorithm must immediately and irrevocably decide what color to assign to each edge $e_t$ as it arrives. The subgraphs induced by each color must form a matching, meaning each vertex can have at most one edge of each color. The objective is to mimimize the number of colors, getting as close as possible to the $\Delta+1$ colors guaranteed by Vizing's Theorem \cite{vizing1964estimate}.

A simple folklore algorithm yields an edge coloring using at most $(2\Delta-1)$ colors.

\begin{lemma}[Greedy Edge Coloring.]
\label{lem:greedy}
There is a deterministic online edge coloring algorithm that uses at most $2\Delta-1$ colors without knowledge of $\Delta$.
\end{lemma}
\begin{proof}
Use a palette $\Cgreedy$, initially empty. When an edge $e=\{u,v\}$ arrives, color it with a color from $\Cgreedy$ if possible. If not, add a fresh color to $\Cgreedy$ and color $\{u,v\}$ using it.
Note that there are at most $\Delta-1$ other edges incident at $u$, and similarly for $v$. So at most $2(\Delta-1)$ colors of $\Cgreedy$ could be unavailable for $r$, meaning that the greedy algorithm will at all times has $|\Cgreedy|\le 2\Delta-1$.
\end{proof}

\paragraph{Online adversaries.}
Deterministic algorithms must work no matter the input and is not allowed to use any randomness. 
If the algorithm is randomized, it is important to specify the \emph{adversary} it works against. An \emph{oblivious} adversary fixes the graph and arrival order in advance.
An \emph{adaptive} adversary is more powerful, and can generate the graph adaptively based on the algorithm's prior random choices.
A classic result asserts that randomization does not help against adaptive adversaries in the online setting: if one can design a randomized algorithm against adaptive adversaries, then there also exists a deterministic algorithm with the same guarantees \cite{ben1994power}.
As our deterministic algorithms in this work follow from the above simple reduction, we will only focus on randomized algorithms, specifying whether we are working in the easier oblivious or adaptive adversary settings. In the adaptive adversary setting, by Yao's minmax principle \cite{yao1977lemma}, we assume without loss of generality that the adversary is deterministic (picking the worst edge given prior choices of the algorithm).

\subsection{Martingales}\label{sec:martingales}
Martingales and their concentration inequalities are key tools in our algorithms' analyses.

\begin{definition}[Supermartingales and Martingales] \label{def:martingales}
    A sequence of random variables $Z_0,Z_1,\dots,Z_m$ is a \emph{supermartingale} with respect to another sequence of random variables $X_1,\dots,X_m$, if
    $Z_{i-1}$ is determined by $X_{1}, \ldots, X_{i-1}$ and
    $$\E[Z_i \mid X_1,\dots,X_{i-1}] \leq Z_{i-1} \qquad \forall i\in [m].\footnote{Strictly speaking, the right-hand-side is also conditional. Specifically, the right-hand side is $[Z_{i-1} \mid X_1,\dots,X_{i-1}]$. As writing out this conditioning explicitly results in cumbersome notation, it is common practice (which we will adopt) to only implicitly use this conditioning when proving that a sequence is a (super)martingale.
    }$$
    If the above inequality is, in fact, always an equality, the sequence is called a \emph{martingale}.
\end{definition}

\begin{lemma}[Azuma's inequality] \label{lemma:azuma}
    Let $Z_0,\dots,Z_m$ be a supermartingale w.r.t.\ the sequence $X_1,\dots,X_m$. Assume the \emph{step size} $A$ is a positive real number such that $|Z_i - Z_{i-1}| \leq A$ for all $i \geq 1$. Then, for any $i \in \{1,\dots,m\}$ and positive real $\lambda$, we have the following tail bound:
    \begin{equation*}
        \Pr[Z_i - Z_0 \geq \lambda] \leq \exp \left( - \frac{\lambda^2}{2iA^2} \right).
    \end{equation*}
    If the sequence is a martingale, then the tail bound in the other direction also holds:
    \begin{equation*}
        \Pr[Z_0 - Z_i \geq \lambda] \leq \exp \left( - \frac{\lambda^2}{2iA^2} \right).
    \end{equation*}
\end{lemma}

\section{Deterministic (via Adaptive) Online Edge Coloring}
\label{sec:adaptive}

In this section we prove that when $\Delta = \omega(\log n)$, one can edge-color a graph online using at most $(1+o(1))\Delta$  colors with a \emph{deterministic} algorithm. More generally, we prove the following.
\begin{theorem}\label{thm:deterministic-precise}
There is a deterministic online algorithm that edge-colors $n$-node graphs with known $n$ and $\Delta$ using
$$
\Delta + O(\Delta^{15/16}\log^{1/16}n) \textrm{ colors}.
$$
\end{theorem}

We prove this
by designing a \emph{randomized} algorithm that works against \emph{adaptive} adversaries, and combining this with the classic reduction of \cite{ben1994power}.

\begin{theorem}
\label{thm:adaptive-theorem}
There is a randomized online algorithm that edge-colors $n$-node graphs with known $n$ and $\Delta$, possibly generated adaptively, using w.h.p.
$$
\Delta + O(\Delta^{15/16}\log^{1/16}n) \textrm{ colors}.
$$
\end{theorem}

Naturally, this section is dedicated to proving \Cref{thm:adaptive-theorem}.

\subsection{The Algorithm}\label{sec:algorithm}
The algorithm, whose pseudo-code is given in \cref{alg:edge-coloring}, is quite simple.
It intuitively attempts to assign each edge a color chosen uniformly at random from the palette $\Calg=[\Delta]$, or failing that, uses the greedy algorithm with a separate backup palette $\Cgreedy$.

How is this achieved concretely? Initially, we assign each possible future edge $e$ and color $c\in \Calg$ a uniform probability $P_{ec}\gets \frac{1-\eps}{\Delta}$.\footnote{See pseudocode or \Cref{def:parameters} for exact value of $\eps$ and $A$ below.}
These probabilities change over time, with their values at time $t$ denoted by $P^{(t)}_{ec}$, used and updated as follows. 
When edge $e_t$ arrives, we try to sample a color $K_t$ from $\Calg$, with each color $c\in \Calg$ picked with probability $P_{e_{t}c}^{(t-1)}$. 
If this is impossible, due to $\sum_{c\in \Calg} P^{(t-1)}_{e_tc}>1$, we \emph{mark} $e_t$ and color it greedily using colors from $\Cgreedy$.
Otherwise, we sample $K_t$ from $\Calg\cup \{\perp\}$, with each color $c$ sampled with probability $P_{e_tc}^{(t-1)}$, and $\perp$ sampled with the remaining probability, $1-\sum_{c\in \Calg} P_{e_tc}^{(t-1)}$. 
If we sample a color $K_t\in \Calg$, then we assign it to $e_t$, otherwise, we \emph{mark} $e_t$ and greedily color it.
We then update all probabilities of possible future edges $f \in F_t$ neighboring $e_t$. First, we zero out the probability $P_{fK_t}^{(t)}$, to prevent $f$ from using the color $K_t$ later. Then, in an attempt to (almost) preserve the $P$-values in expectation, we set $P^{(t)}_{fc}\gets {P^{(t-1)}_{fc}}/{(1-P_{e_tc}^{(t-1)})}$ for all other colors $c\in \Calg\setminus\{K_t\}$, but only if $P_{fc}^{(t-1)}$ is not too large, specifically $P^{(t-1)}_{fc}\leq A$ for a threshold $A$. For all other colors and possible future edges (where we denote the possible future edges at time $t$ by $F_t$), we have $P^{(t)}_{ec} \gets P^{(t-1)}_{ec}$.

\begin{algorithm}[ht!]
	\caption{Edge Coloring Algorithm Against Adaptive Adversaries}
	\label{alg:edge-coloring}
	\begin{algorithmic}[1]
        \Statex \underline{\smash{\textbf{Input:}}} Vertex set $V$ and maximum degree $\Delta\in \mathbb{Z}_{\geq 0}$ of the graph to arrive
        \Statex 
                      \textbf{\underline{Initialization:}} 
        $\C_{\mathrm{alg}} \leftarrow [\Delta]$, 
        $F_0 \leftarrow \binom{V}{2}$,
        $\eps \gets c_{\eps}\cdot(\frac{\ln n}{\Delta})^{1/16}$,
        $A \gets \frac{c_{A}}{\eps^2\Delta}$, where $c_\varepsilon := 10$ and $c_A := 4$.
        \Statex\qquad\qquad\qquad\;\;\textbf{for each} $e \in F_0$ and $c \in \Calg$:\; Set $P^{(0)}_{ec} \leftarrow \frac{1 - \varepsilon}{\Delta}$. 
        
        \Statex 
		\For{\textbf{each} online edge $e_t$ on arrival} 
            \State $F_t \leftarrow F_{t-1} \setminus \{e_t\}$.
            \State \textbf{for each} $f\in \binom{V}{2}$ and $c\in \Calg$: \; Set $P^{(t)}_{fc}\gets P^{(t-1)}_{fc}$. \Comment{May be overridden below if $f$ neighbors $e_t$}
            \If{$\sum_{c\in \Calg} P^{(t - 1)}_{e_tc} > 1$}
                Mark $e_t$.
                \label{line:mark_z_ge_1}
            \Else
                \State Sample $K_t$ from $\Calg \cup \{\perp\}$ with probabilities $\left(P^{(t-1)}_{e_t1},\dots,P^{(t-1)}_{e_t\Delta}, 1 - \sum_{c\in \Calg}P^{(t-1)}_{e_tc}\right)$.
                \label{line:sample}
                \If{$K_t\in \Calg$} Assign color $K_t$ to $e_t$. \label{line:color}
                \Else\, Mark $e_t$. \label{line:mark_bot}
                \EndIf
                \For{\textbf{each} $f \in F_t$ such that $f \cap e_t \neq \emptyset$}
                    \State $P^{(t)}_{fK_t} \leftarrow 0$ \textbf{if} $K_t \in \Calg$
                        \Comment{Prevent $f$ from being colored $K_t$}
                    \label{line:zero}
                        \State $P^{(t)}_{fc} \leftarrow\frac{P^{(t-1)}_{fc}}{1 - P^{(t-1)}_{e_tc}}$
                        \textbf{for each}
                    $c\in \Calg\setminus\{ K_t\}$ \textbf{where} {\color{blue} $P^{(t-1)}_{fc} \leq A$ ($\star$)}
                    \label{line:scaleup}
                    \label{line:star}
                \EndFor
            \EndIf

            \If{$e_t$ marked}
            \State Color $e_t$ using the greedy algorithm (\cref{lem:greedy}) with the separate palette $\Cgreedy$.
                \label{line:greedy}
            \EndIf
		\EndFor
	\end{algorithmic}
\end{algorithm}

It is easy to see that the algorithm outputs a valid edge coloring.
\begin{lemma} \label{lemma:valid_coloring}
    \Cref{alg:edge-coloring} produces a valid edge coloring on any input graph, even under adaptive edge arrivals.
\end{lemma}
\begin{proof}
    Each edge that is marked is colored in \Cref{line:greedy} and assigned a different color than its neighboring edges, by correctness of the greedy algorithm (\cref{lem:greedy}) and disjointness of $\Calg$ and $\Cgreedy$. In contrast, whenever an edge $e_t$ is not marked and is thus colored using a color $K_t\in \Calg$ in \Cref{line:color}, we set the probabilities $P_{fK_t}^{(t)}$ of all possible future neighboring edges $f$ to zero (and never increase this $P$ value), meaning these edges cannot be assigned this color later, as this color will not be sampled for $f$ in \Cref{line:sample}. So, no two neighboring edges share a common color in $\Calg$, and \Cref{alg:edge-coloring} produces a valid edge~coloring.
\end{proof}

\Cref{alg:edge-coloring} uses colors from $\Calg$ and the ``emergency'' palette $\Cgreedy$, leading to a total of $|\Calg|+|\Cgreedy| = \Delta + |\Cgreedy|$ colors.
By \cref{lem:greedy}, we have that $|\Cgreedy|\le 2\Delta'$, where $\Delta'$ is the maximum degree in the graph of \emph{marked} edges, we conclude that \Cref{alg:edge-coloring} uses at most $3\Delta$ colors. However, this bound turns out to be very loose (if $\Delta = \omega(\log n)$). In this section, we will prove that, with high probability, each vertex is incident to at most $O(\eps \Delta)$ = $o(\Delta)$ \emph{marked} edges. Thus, $2\Delta' = o(\Delta)$, and the algorithm in fact uses (with high probability) less than $\Delta+O(\Delta')=\Delta(1+o(1))$ colors. 

Before proceeding with the analysis overview, we make brief observations about \cref{alg:edge-coloring} and its parameters which will prove useful in our analysis.

\begin{definition}[Parameters] \label{def:parameters}
    Let $\varepsilon := c_\varepsilon \cdot \left( \frac{\ln n}{\Delta} \right)^{1/16}$,\footnote{All our logarithms in this paper are natural logarithms (base $e$).}  and $A := \frac{c_A}{\varepsilon^2 \Delta}$, for constants $c_\varepsilon := 10$ and $c_A := 4$.
\end{definition}
Our target is to show that the algorithm uses $(1+O(\eps)) \Delta = \Delta + O(\Delta^{15/16}\log^{1/16} n)$ colors.
    When $\Delta < (10c_{\eps})^{16} \ln n$, this is trivial, as the algorithm never uses more than $3\Delta$ colors, which is at most $\Delta + O(\Delta^{15/16}\log^{1/16} n)$ in this case. In the rest of this section we thus assume, without loss of generality, that $\Delta \ge (10c_{\eps})^{16} \ln n$, implying that $\eps$ and $A$ are both small.

\begin{fact}
\label{fact:delta_large_enough}
    When $\Delta \ge (10c_{\eps})^{16} \ln n$, then $\eps \le \frac{1}{10}$ and $A \le \frac{1}{10}$.
\end{fact}

Next, we note that throughout the algorithm's execution, the $P$-values are always small.
\begin{fact}
\label{lemma:invariant_ub_on_P}
    For any time step $t$, potential future edge $f \in F_t$ and color $c \in \Calg$, it holds that $P^{(t)}_{fc} \leq 2A$.
\end{fact}
\begin{proof}
    By induction on $t\geq 0$. The base case holds trivially for $t=0$, for which $P^{(0)}_{fc} = \frac{1-\eps}{\Delta} \leq \frac{2c_A}{\eps^2 \Delta} =  2A$, using that $\eps\leq 1\leq c_A$ by \Cref{fact:delta_large_enough}. For the inductive step, for edge-color pairs $(f,c)$ for which $P^{(t)}_{fc}=P^{(t-1)}_{fc}$, this follows by the inductive hypothesis. For other edge-color pairs $(f,c)$ we must have that $f$ neighbors $e_t$ and so either we set $P^{(t)}_{fc}=0\leq 2A$ in \Cref{line:zero}, or we set this value in \Cref{line:scaleup}, where we require $P^{(t-1)}_{fc}\leq A$. By the inductive hypothesis and that $A\leq \frac{1}{4}$ by \Cref{fact:delta_large_enough}, this case too yields the desired bound:
    \begin{align*}
    P^{(t)}_{fc} & = \frac{P^{(t-1)}_{fc}}{1 - P^{(t-1)}_{e_tc}} \leq \frac{A}{1 - 2A} \leq 2A. \qedhere 
    \end{align*}
\end{proof}

\subsection{Analysis} \label{sec:analysis}

\paragraph{Analysis Overview.} The probability of an arriving edge $e_t$ being successfully colored using the first palette $\Calg$ in \cref{alg:edge-coloring} is determined by
$$Z^{(t-1)}_{e_t}:=\sum_c P^{(t-1)}_{e_tc}.$$
Indeed, if $Z^{(t-1)}_{e_t} \leq 1$, then we pick some color $K_t \in \Calg$ to color $e_t$ with probability precisely $Z^{(t-1)}_{e_t}$. Therefore, if we prove that, with high probability, for all time steps $t$ we have both $Z^{(t-1)}_{e_t} \leq 1$ and $Z^{(t-1)}_{e_t} > 1 - O(\eps)$, then (intuitively, and formalized via a coupling argument in \Cref{lemma:decrease_of_deg_at_v} and Chernoff bounds) this will result w.h.p.~in all but $O(\eps\Delta)$ many edges $e_t$ of each vertex being successfully colored without needing to resort to $\Cgreedy$. 
Thus, w.h.p.~the ``emergency'' palette $\Cgreedy$ remains reasonably small, and so \cref{alg:edge-coloring} uses $(1 + O(\eps)) \Delta = (1+o(1))\Delta$ colors with high probability.

To argue that each $Z_{e_t}^{(t-1)}$ is likely in the range $[1-O(\eps),1]$, we study the sequence $Z^{(0)}_{f}, Z^{(1)}_{f},\ldots$ for an arbitrary edge $f \in F_0$. We prove that each such sequence is a supermartingale of bounded step size with respect to the samples $K_1,K_2,\dots$, even against an adaptive adversary.\footnote{Recall that the adversary can be assumed to be deterministic without loss of generality, and so the only randomness in the process is given by the samples $K_1,K_2,\dots$.} This follows from linearity and the fact that each edge-color probability sequence, $P^{(0)}_{fc},P^{(1)}_{fc},\dots$, is itself a supermartingale. Moreover, had we not imposed the $\starcond$-condition $P^{(t-1)}_{fc}\leq A$ for increasing $P^{(t)}_{fc}$ in \Cref{line:scaleup}, the sequence would in fact even be a martingale, though it would not have bounded step size. 

The following lemma proves that $Z^{(0)}_{f}, Z^{(1)}_{f},\ldots$ is indeed a supermartingale of small step size. While it is implied by  \Cref{lemma:Zlbhelperlemma1}, we provide a proof of this lemma and use it as a simple warm-up and template for subsequent analyses of (super)martingales we use.

\begin{lemma} \label{lemma:Z_is_supermartingale}
    Let $f \in F_0$. The sequence of random variables $Z^{(0)}_{f}, Z^{(1)}_{f}, \dots$ is a supermartingale with respect to the sequence of random variables $K_1, K_2, \dots$. Furthermore, the step size is bounded by $6A$.
\end{lemma}
\begin{proof}
    Fix $t \geq 0$ and consider the $(t+1)$-th time step. If $f \notin F_{t+1}$ (i.e., $f$ already arrived by time $t+1$) or $Z^{(t)}_{e_{t+1}} > 1$ (leading to $e_{t+1}$ being \emph{marked} in \cref{line:mark_z_ge_1}), or $f \cap e_{t+1} = \emptyset$, then we have $Z^{(t+1)}_f = Z^{(t)}_f$, and so there is nothing to prove. So, assume $f \in F_{t+1}$, and  $Z^{(t)}_{e_{t+1}} \leq 1$, and $f\cap e_{t+1}\neq \emptyset$.
    Letting  $B = \{c \in \Calg: P^{(t)}_{fc} > A\}$, using $K_{\leq t}$ as shorthand for $K_1,\dots,K_t$, and inspecting the updates to $P$ in \cref{line:zero,line:scaleup}, we have:
    {
    \begin{align*}
        &\E\left[Z^{(t+1)}_{f} - Z^{(t)}_f \;\middle\vert\; K_1,\dots,K_t \right] = \sum_{c \in \Calg} \E\left[ P^{(t+1)}_{fc} - P^{(t)}_{fc}\;\middle\vert\; K_1,\dots,K_t \right]\\
        &= \sum_{c \in \Calg} \left[ \Pr\left[K_{t+1} = c \mid K_{\le t}\right] \cdot \left( -P^{(t)}_{fc} \right) + \mathds{1}[c \notin B] \cdot \left( 1 - \Pr\left[K_{t+1} = c\mid K_{\le t} \right] \right) \cdot \left( \frac{P^{(t)}_{fc}}{1 - P^{(t)}_{e_{t+1}c}} - P^{(t)}_{fc} \right) \right] \\
        &\leq \sum_{c \in \Calg} \left[ P^{(t)}_{e_{t+1}c} \cdot \left( -P^{(t)}_{fc} \right) + \left( 1 - P^{(t)}_{e_{t+1}c} \right) \cdot \left( \frac{P^{(t)}_{fc}}{1 - P^{(t)}_{e_{t+1}c}} - P^{(t)}_{fc} \right) \right] \\
        &= 0.
    \end{align*}}
    This proves that the sequence $Z^{(0)}_{f}, Z^{(1)}_{f}, \dots$ is a supermartingale. It remains to bound the step size:
    \begin{align*}
        \left| Z^{(t+1)}_{f} - Z^{(t)}_f \right| &\leq \left|-\mathds{1}[K_{t+1} \neq \perp] \cdot P^{(t)}_{fK_{t+1}} + \sum_{c \in \Calg \setminus (B \cup \{K_{t+1}\})} \left( \frac{P^{(t)}_{fc}}{1 - P^{(t)}_{e_{t+1}c}} - P^{(t)}_{fc}\right)\right| \\
        &\leq \mathds{1}[K_{t+1} \neq \perp] \cdot P^{(t)}_{fK_{t+1}} + \sum_{c \in \Calg} \left( \frac{P^{(t)}_{fc}}{1 - P^{(t)}_{e_{t+1}c}} - P^{(t)}_{fc}\right) \\
        &\leq 2A + \sum_{c \in \Calg} \frac{P^{(t)}_{fc} \cdot P^{(t)}_{e_{t+1}c}}{1 - P^{(t)}_{e_{t+1}c}} \\
        &\leq 2A + 2\cdot \sum_{c \in \Calg} P^{(t)}_{fc} \cdot P^{(t)}_{e_{t+1}c}.
    \end{align*}
    Above, we first make use of the triangle inequality, and increase the second sum by adding the terms corresponding to the colors from $B \cup \{K_{t+1}\}$ (which were excluded before).  Then, we use the invariant from \cref{lemma:invariant_ub_on_P} to upper bound the first term by $2A$, and then use the invariant again to argue $1/(1 - P^{(t)}_{e_{t+1}c}) \leq 1/(1-2A) \leq 1/2$.

    It remains to argue that $\sum_{c \in \Calg} P^{(t)}_{fc} \cdot P^{(t)}_{e_{t+1}c} \leq 2A$, since that will imply $\left| Z^{(t+1)}_{f} - Z^{(t)}_f \right| \leq 6A$, which is what we need to prove. To see this fact, first note that $Z^{(t)}_{e_{t+1}} = \sum_{c \in \Calg} P^{(t)}_{e_{t+1}c} \leq 1$, by the assumption made at the start of the proof.
    Hence, again by \cref{lemma:invariant_ub_on_P}, 
    \begin{align*}
        \sum_{c \in \Calg} P^{(t)}_{fc} \cdot P^{(t)}_{e_{t+1}c} & \leq \max\left\{P^{(t)}_{fc} : c \in \Calg\right\} \leq 2A. \qedhere 
    \end{align*}
\end{proof}

That $Z^{(0)}_{f},Z^{(1)}_{f},\dots$ is a supermartingale is sufficient to argue that $Z_{e_t}^{(t-1)}$ does not deviate much \emph{above} $Z_{e_t}^{(0)} = 1-\eps$ with high probability  (using Azuma's inequality, as applied in the later \cref{lemma:bounds_on_Z}). However, we also need to bound the deviation \emph{below} $1-\eps$. Continuing from the preceding discussion, if we were to drop the $\starcond$-condition $P^{(t-1)}_{fc}\leq A$ in \Cref{line:scaleup}, the sequences $P^{(0)}_{fc},P^{(1)}_{fc},\dots$, and hence $Z^{(0)}_{f},Z^{(1)}_{f},\dots$, would be \emph{martingales} (and not just supermartingales) with respect to $K_1,K_2,\dots$. 
This would suffice to argue the desired statement that $Z_{e_t}^{(t-1)}$ is unlikely to deviate much \emph{below} $Z_{e_t}^{(0)}=1-\eps$. 

Motivated by the above, we define a sequence $Y$ that considers the per-step change \emph{if we were} to allow to (temporarily) scale up the $P$-values ignoring the $\starcond$-condition $P^{(t-1)}_{fc} \leq A$. We will formalize the discussion in the previous paragraph by showing that $Y$ is a martingale. Furthermore, if there were no ``bad'' edge-color pairs which violate the $\starcond$-condition, then this martingale $Y$ would be exactly $Z$. Intuitively, and as we formalize shortly, if there are few such pairs, $Y$ and $Z$ are almost equal. This allows us to lower bound $Z$ with high probability in terms of the martingale $Y$, which is sharply concentrated. $Y$ is defined as follows:
\begin{definition} \label{def:lb_on_Z}
For each $f\in F_{0}$, color $c\in \Calg$, and time step $t$, let
\begin{equation*}
\overline{P^{(t)}_{fc}}
=
\begin{cases}
P^{(t-1)}_{fc} / (1-P^{(t-1)}_{e_tc}) & \text{if in \Cref{line:scaleup} no assignment to $P^{(t)}_{fc}$ is made because of ${\color{blue}P^{(t-1)}_{fc} > A\quad (\star)}$} \\
P^{(t)}_{fc} & \text{otherwise.}
\end{cases}
\end{equation*}
Also define:
    \begin{equation*}
        \overline{Z^{(t)}_{e}} = \sum_{c \in \Calg} \overline{P^{(t)}_{ec}} \text{ \ \ and \ \ } Y^{(t)}_e = \sum_{t'=1}^t \left( \overline{Z^{(t')}_{e}} - Z^{(t'-1)}_e \right).
    \end{equation*}
\end{definition}
Next, we prove that the sequence $Y$ is a martingale with bounded step size, similarly to the proof of \Cref{lemma:Z_is_supermartingale} (see \Cref{sec:proofs_Zlbhelperlemmas}).
\begin{restatable}{lemma}{Zlbhelperlemmaone} \label{lemma:Zlbhelperlemma1}
    Let $f \in F_0$. The sequence of random variables $Y^{(0)}_f, Y^{(1)}_f,\dots$ is a martingale with respect to the sequence of random variables $K_1,K_2,\dots$ Furthermore, the step size is bounded by $6A$.
\end{restatable}
Next, we use the sequence $Y$ to represent $Z$ as the sum between a martingale and a ``drift''. We recall that $Z_e^{(0)}=1-\eps$, and by telescoping can thus be written as $Z^{(t)}_e = 1-\eps + \sum_{t'=1}^{t} \left( Z^{(t')}_e -  Z^{(t'-1)}_e \right)$. Rewriting this in terms of $Y$, we obtain the following decomposition of $Z$:
\begin{equation}\label{eqn:decomposition}
    Z^{(t)}_e  = 1 - \varepsilon + Y^{(t)}_e - \sum_{t'=1}^{t} \left( \overline{Z^{(t')}_e} -  Z^{(t')}_e \right).
\end{equation}
Note that since $\overline{Z^{(t)}_e} \geq Z^{(t)}_e$ always, this implies that $Z^{(t)}_e \leq 1-\eps + Y^{(t)}_e$, and so \Cref{lemma:Zlbhelperlemma1} together with Azuma's inequality give an alternative proof of the upper tail bound for $Z^{(t)}$ which was proven directly in \cref{lemma:Z_is_supermartingale}, using the fact that $Z^{(t)}_e$ is a supermartingale. More importantly, the decomposition above implies that $Z^{(t)}_e$ differs from $Y^{(t)}_e +1-\eps$ (which is sharply concentrated around $1-\eps$, using Azuma's inequality) by a sum of bounded terms (by \Cref{lemma:invariant_ub_on_P}). Moreover, these terms are only non-zero for those timesteps $t'$ and those colors $c$ such that $c$ is \emph{bad} at time step $t'$, meaning that $\overline{P^{(t')}_{ec}}\neq P^{(t')}_{ec}$.
This motivates the following lemma, which together with the preceding discussion, will allow us to lower bound $Z$ with high probability:
\begin{restatable}[Few Bad Colors Lemma]{lemma}{MainLemma} \label{lemma:main_lemma}
    With probability at least $1-1/n^{100}$, at all time steps $t$ during the execution of \cref{alg:edge-coloring}, and for all edges $f \in F_t$, we have:
    \begin{equation*}
        \left| B^{(t)}_f \right| \leq 2\varepsilon^5 \Delta,
    \end{equation*}
    where $B^{(t)}_f := \left\{c \in \Calg: P^{(t)}_{fc} > A \right\}$ is called the set of \emph{bad} colors w.r.t.\ the potential future edge $f$ at time $t$.
\end{restatable}

The above lemma is the technical meat of our analysis of \Cref{alg:edge-coloring}, and we  dedicate \cref{sec:arguing_main_lemma} to its proof. 
For now, we note that together with the preceding discussion, it implies the following lower tail bound on $Z$ quite directly (see \cref{sec:proofs_Zlbhelperlemmas} for a full proof). 

\begin{restatable}{lemma}{Zlbhelperlemmatwo} \label{lemma:Zlbhelperlemma2}
    Let $f \in F_0$ and $t$ be an arbitrary time step. We have, with probability at least $1-1/n^{100}$:
    \begin{equation*}
        Z^{(t)}_{f} \geq 1 - \varepsilon + Y^{(t)}_f - 20 \varepsilon^5 \Delta^2 A^2.
    \end{equation*}
\end{restatable}

We proceed by leveraging the (super)martingales $Z$ and $Y$ to prove our desired w.h.p.~bounds on $Z^{(t)}_e$. Recall that these bounds will help us later argue that only few edges per vertex are marked and forwarded to greedy, and so the ``emergency'' palette $\Cgreedy$ will be small.
\begin{lemma} \label{lemma:bounds_on_Z}
    Fix a potential future edge $e \in F_0$, and consider the associated supermartingale $Z^{(0)}_{e}, Z^{(1)}_{e},\dots$. With probability at least $1 - 1/n^{99}$, for any time step $t$ such that $e \in F_t$, we have:
    \begin{equation*}
        1 - 2 \varepsilon - 20 \varepsilon^5 \Delta^2 A^2 \leq Z^{(t)}_e \leq 1.
    \end{equation*}
\end{lemma}
\begin{proof}
    We begin by proving the upper bound. Since $Z^{(0)}_e = 1-\varepsilon$, it suffices to prove:
    \begin{equation*}
        \Pr\left[ Z^{(t)}_e - Z^{(0)}_e > \varepsilon \right] \leq \frac{1}{n^{100}}.
    \end{equation*}
    To set up Azuma's inequality, let $\lambda = \varepsilon$ and $t^+ = \left|\left\{ t' < t - 1: Z^{(t'+1)}_e \neq Z^{(t')}_e \right\} \right|$. These are the number of time steps that lead to (non-zero) changes in the supermartingale $Z_e$. Apply \cref{lemma:azuma} to obtain:
    \begin{equation*}
        \Pr\left[ Z^{(t)}_e - Z^{(0)}_e > \varepsilon \right] \leq \exp\left( -\frac{\varepsilon^2}{2t^+(6A)^2} \right).
    \end{equation*}
    It suffices to argue that:
    \begin{equation*}
        \frac{\varepsilon^2}{2t^+(6A)^2} \geq 100 \ln n.
    \end{equation*}
    Note that $t^+ \leq 2\Delta$, since $Z_e$ can only change due to online arrivals of edges $f$ such that $f \cap e \neq \emptyset$. Since $\Delta$ is the maximum degree in the graph, there are at most $2\Delta$ such edges. 
    Hence, it suffices to argue that:
    \begin{equation*}
        \frac{\varepsilon^2}{144\Delta A^2} \geq 100 \ln n.
    \end{equation*}
    Plugging in $A = c_A/(\varepsilon^2\Delta)$, the above inequality can be simplified to:
    \begin{equation*}
        \frac{\varepsilon^6 \Delta}{(12c_A)^2} \geq 100 \ln n.
    \end{equation*}
    Since $\varepsilon^6 \geq \varepsilon^{16} = c^{16}_\varepsilon \cdot (\ln n/\Delta)$, it suffices to  argue that:
    \begin{equation*}
        \frac{c_\varepsilon^{16}}{(12c_A)^2} \geq 100,
    \end{equation*}
    which holds by the choices of the constants $c_\varepsilon$ and $c_A$.

    To prove the lower bound, thanks to \cref{lemma:Zlbhelperlemma2} (which says $Z_{e}\ge 1-\eps + Y_{e}-20\eps^5\Delta^2A^2$ holds with probability at least $1-1/n^{100}$),
    it suffices to argue that $Y^{(t)}_{e} > -\varepsilon$ with high probability. Since $Y^{(0)}_e = 0$, it suffices to prove:
    \begin{equation*}
        \Pr\left[ Y^{(0)}_e - Y^{(t)}_{e} > \varepsilon \right] \leq \frac{1}{n^{100}}.
    \end{equation*}
    This can be done again by applying Azuma's inequality (\cref{lemma:azuma}) in an identical way to above, using \Cref{lemma:Zlbhelperlemma1}.
\end{proof}
\begin{corollary} \label{corollary:Z_is_good_whp_always}
    Let $\text{Zgood}$ be the event that, in an execution of \cref{alg:edge-coloring}, for any time step $t$, and any potential future edge $e \in F_t$, it holds that $1 - 2 \varepsilon - 20  \varepsilon^5 \Delta^2 A^2 \leq Z^{(t)}_e \leq 1$. We call such executions of \cref{alg:edge-coloring} \emph{good}. Then:
    \begin{equation*}
        \Pr[\mathrm{Zgood}] \geq 1 - \frac{1}{n^{90}}.
    \end{equation*}
\end{corollary}
\begin{proof}
    This follows from \cref{lemma:bounds_on_Z} by taking an union bound over all $\leq n^2$ time steps and $\leq n^2$ potential future edges.
\end{proof}

The above corollary allows us to focus on \emph{good} executions, i.e., executions for which $Z_{e}$ is always well concentrated. In this case, all vertices $v$ in the input graph have all but a small fraction of their edges successfully colored using the initial palette $\Calg$. To obtain a formal argument of this statement, we first show that each edge incident at a vertex $v$ is almost surely colored using $\Calg$ (equivalently, it is not \emph{marked} by \cref{alg:edge-coloring}), independently of what happened with previous edges incident at $v$. Together with a simple coupling argument, this fact allows us to apply Chernoff bounds to bound with high probability the number of uncolored edges of $v$.

{
\begin{lemma} \label{lemma:decrease_of_deg_at_v}
    Let $v \in V$ be a fixed vertex in the input graph, and consider its incident edges $e_{t_1} < e_{t_2} < \dots < e_{t_k}$ in order of their arrival $t_i$. Let $M_v$ be the number of edges incident at $v$ which are marked by the algorithm. We have that:
    \begin{equation*}
        \Pr\left[ M_v \geq c \cdot (\varepsilon \Delta) + \sqrt{300 \cdot \Delta \ln n} \right] \leq \frac{1}{n^{85}},
    \end{equation*}
    where $c := 2 + 20c_A$ is a constant.
\end{lemma}
\begin{proof}

For edge $e_{t_i}$, define an indicator random variable $X_{i}$ that is $1$ if and only if the edge was marked in \cref{line:mark_bot} \emph{and} $Z^{(t_{i}-1)}_{e_{t_i}} \in [1-c\cdot \eps, 1]$ (and otherwise $X_{i} = 0$). Note that $c\cdot \eps = 2 + 20c_A\eps = 2 + 20 \eps^{5}\Delta^2A^2$.

By definition, $\Pr[X_{i+1} \mid X_{1}, X_{2}, \ldots, X_{i}] \le c\cdot \eps$, as, no matter the history, the edge will be marked in \cref{line:mark_bot} with probability $1-Z^{(t_i - 1)}_{e_{t_i}}$. Since $v$ has degree at most $\Delta$, we have $\E[\sum_{i} X_i] \le c\eps\Delta$. Via standard coupling arguments $\sum_{i} X_i$ is upper bounded by the sum of $\Delta$ independent Bernoulli$(c\eps)$ random variables, and we can apply the following Chernoff-Hoeffding bound for $\delta := \sqrt{300 \cdot \Delta \ln n}$:
    \begin{align*}
        \Pr\left[ \sum_{i} X_i \geq c\eps \Delta + \delta \right]  
        \leq \exp \left( -\frac{2\delta^2}{\Delta} \right)
        \le \frac{1}{n^{100}}.
    \end{align*}

We finish the proof by noting that the number of marked edges $M_v$ incident to $v$ is at most $\sum_{i} X_i$ plus the number of edges $e_{t_{i}}$ where $Z^{(t_i - 1)}_{e_{t_i}}\notin [1-c\cdot\eps, 1]$, and by \cref{corollary:Z_is_good_whp_always} the latter term is $0$ with probability $1-1/n^{90}$ (so we can union bound over this event and the Chernoff bound above being concentrated).
\end{proof}
}

\begin{corollary} \label{corollary:v_is_good_whp_always}
    Let $\text{DegGood}$ be the event that
    all vertices $v$ have at most $c \cdot (\varepsilon \Delta) + \sqrt{300 \cdot \Delta \ln n}$ incident marked edges. We have that:
    \begin{equation*}
        \Pr\left[ \mathrm{DegGood} \right] \geq 1 - \frac{1}{n^{80}}.
    \end{equation*}
\end{corollary}
\begin{proof}
    This follows from \cref{lemma:decrease_of_deg_at_v} by taking an union bound over all $\leq n$ vertices.
\end{proof}

Given that an execution of \cref{alg:edge-coloring} is \emph{good} with high probability, and that \emph{good} executions guarantee a decrease to $o(\Delta)$ $\Calg$-uncolored edges per vertex with high probability, we can now argue that \cref{alg:edge-coloring} will use at most $o(\Delta)$ colors from the ``emergency'' palette $\Cgreedy$. The below lemma finishes our proof of \cref{thm:adaptive-theorem}.
\begin{lemma}
    With probability at least $1 - 1/n^{70}$, \cref{alg:edge-coloring} colors all edges of the input graph using
    \begin{equation*}
        2c \cdot (\varepsilon \Delta) + 2 \cdot \sqrt{300 \Delta \ln n} = O(\eps \Delta) = O\left(\Delta^{15/16} \log^{1/16}n\right)
    \end{equation*}
    colors from $\Cgreedy$.
\end{lemma}
\begin{proof}
We first note that $2\sqrt{300\Delta \log n} = O(\eps \Delta)$ whenever $\Delta = \Omega(\log n)$ by definition.
    Let $G' \subseteq G$ be the subgraph of the input graph $G$ which contains all marked edges. In a \emph{good} execution of \cref{alg:edge-coloring}, \cref{corollary:v_is_good_whp_always} shows that the maximum degree of $G'$ is, with probability at least $1 - 1/n^{80}$, bounded by $\Delta' = c \cdot (\varepsilon \Delta) + \sqrt{300 \Delta \ln n}$. Since this graph will be colored greedily using the palette $\Cgreedy$ by \cref{alg:edge-coloring}, it easily follows that $\Cgreedy$ contains at most $2\Delta'$ colors.

    Hence, the only cases in which \cref{alg:edge-coloring} potentially uses more than $\Delta + 2\Delta'$ colors is if the execution is not \emph{good} or the execution is \emph{good} but $\text{DegGood}$ fails to hold. From \cref{corollary:Z_is_good_whp_always} and \cref{corollary:v_is_good_whp_always} these events can happen with probability at most $1/n^{100} + 1/n^{80} \leq 1/n^{70}$. This gives the statement of the theorem.
\end{proof}

\subsection{Bounding the Drift: Proofs of \cref{lemma:Zlbhelperlemma1,lemma:Zlbhelperlemma2}} \label{sec:proofs_Zlbhelperlemmas}

For convenience, we reproduce the statements next to the proofs of the two relevant lemmas:
\Zlbhelperlemmaone*
\begin{proof}
    The proof is closely related to the proof of \cref{lemma:Z_is_supermartingale}. Fix $t \geq 0$ and consider the $(t+1)$-th time step. Let $f \in F_{t+1}$ be without loss of generality, for otherwise $\overline{Z^{(t+1)}_f} = Z^{(t)}_f$ trivially. Further, $\overline{Z_f}$ can only change if $f \cap e_{t+1} \neq \emptyset$ (where $e_{t+1}$ is the edge arriving at time $t+1$) and $Z^{(t)}_{e_{t+1}} \leq 1$, so we can also assume w.l.o.g.\ that this is the case. If $K_{t+1} \in \Calg \cup \{\perp\}$ is the color sampled by \cref{alg:edge-coloring} in \cref{line:sample}, we have:
    \begin{align*}
        &\E\left[ Y^{(t+1)}_f - Y^{(t)}_f \;\middle\vert\; K_1,\dots,K_t \right] = \E\left[\overline{Z^{(t+1)}_{f}} - Z^{(t)}_f \;\middle\vert\; K_1,\dots,K_t \right] = \\
        &= \sum_{c \in \Calg} \left[ \Pr\left[K_{t+1} = c \mid K_{\le t} \right] \cdot \left( -P^{(t)}_{fc} \right) + \left( 1 - \Pr\left[K_{t+1} = c\mid K_{\le t} \right] \right) \cdot \left( \frac{P^{(t)}_{fc}}{1 - P^{(t)}_{e_{t+1}c}} - P^{(t)}_{fc} \right) \right] \\
        &= \sum_{c \in \Calg} \left[ P^{(t)}_{e_{t+1}c} \cdot \left( -P^{(t)}_{fc} \right) + \left( 1 - P^{(t)}_{e_{t+1}c} \right) \cdot \left( \frac{P^{(t)}_{fc}}{1 - P^{(t)}_{e_{t+1}c}} - P^{(t)}_{fc} \right) \right]
    \end{align*}
    This is identical to the sum obtained in the proof of \cref{lemma:Z_is_supermartingale} (in that case, as an upper bound on the expectation, as there the second term excluded all bad colors), and thus it equals $0$. The bound on the step size is also identical.
\end{proof}

\Zlbhelperlemmatwo*
\begin{proof}
    Note that (using $Z^{(0)}_{f} = 1 -\varepsilon$), and as observed in \Cref{eqn:decomposition}:
    \begin{equation*}
        Z^{(t)}_f = 1 - \varepsilon + \sum_{t' = 0}^{t-1} \left(Z^{(t' + 1)}_f - Z^{(t')}_f\right) = 1 - \varepsilon + Y^{(t)}_f - \sum_{t'=0}^{t-1} \left( \overline{Z^{(t' + 1)}_f} -  Z^{(t' + 1)}_f \right),
    \end{equation*}
    such that it suffices to upper bound the sum:
    \begin{equation*} 
        S_{\overline{Z}, Z} = \sum_{t'=0}^{t-1} \left( \overline{Z^{(t' + 1)}_f} -  Z^{(t' + 1)}_f \right)
    \end{equation*}
    on the right hand side by $20 \cdot \varepsilon^5 \Delta^2 A^2$. Let $t' \geq 0$ and consider the $(t' + 1)$-th time step. Let $f \in F_{t'+1}$ without loss of generality, for otherwise $\overline{Z^{(t' + 1)}_f} = Z^{(t' + 1)}_f = Z^{(t')}_{f}$ trivially. Further, the difference $\overline{Z^{(t' + 1)}_f} -  Z^{(t' + 1)}_f$ can be non-zero only if $f \cap e_{t'+1} \neq \emptyset$ (where $e_{t'+1}$ is the edge arriving at time $t' + 1$), so we can also assume w.l.o.g.\ that this is the case. Since there are at most $2\Delta$ edges incident to $f$, in particular we have that:
    \begin{equation*}
        \left| \left\{ t': \overline{Z^{(t' + 1)}_f} -  Z^{(t' + 1)}_f \neq 0 \right\} \right| \leq 2\Delta.
    \end{equation*}
    Hence, if we prove an unconditional upper bound $\zeta$ of $\overline{Z^{(t' + 1)}_f} -  Z^{(t' + 1)}_f$ for an arbitrary index $t'$ such that this difference is non-zero, we will obtain $S_{\overline{Z}, Z} \leq 2\Delta \zeta$. Hence, it suffices to argue for such $t'$ that $\overline{Z^{(t' + 1)}_f} -  Z^{(t' + 1)}_f \leq 10 \cdot \varepsilon^5 \Delta A^2$. We fix such a $t'$, and define $B := \{ c \in \Calg : P^{(t')}_{fc} > A \}$ as the set of bad colors for $f$ at time $t'$. By \cref{lemma:main_lemma} we have $|B| \leq 2 \varepsilon^5 \Delta$ with high probability. Further, we have that:
    \begin{align*}
        \overline{Z^{(t' + 1)}_f} -  Z^{(t' + 1)}_f &\leq \sum_{c \in B} \left( \frac{P^{(t')}_{fc}}{1 - P^{(t')}_{e_{t'+1}c}} - P^{(t')}_{fc} \right) \\
        &= \sum_{c \in B} \frac{P^{(t')}_{fc}P^{(t')}_{e_{t'+1}c}}{1 - P^{(t')}_{e_{t'+1}c}} \\
        & \leq \sum_{c \in B} 5A^2 \\
        & \leq 10 \cdot \varepsilon^5 \Delta A^2.
    \end{align*}
    The second inequality above follows from the facts that $P^{(t')}_{fc}, P^{(t')}_{e_{t'+1}c} \leq 2A$ (by \cref{lemma:invariant_ub_on_P}), and the fact that $1 - P^{(t')}_{e_{t'+1}c} \geq 1 - 2A \geq 4/5$ (\cref{fact:delta_large_enough}). The final inequality uses the already established fact that $|B| \leq 2 \varepsilon^5 \Delta$.
\end{proof}

\subsection{Few Bad Colors: Proof of \Cref{lemma:main_lemma}} \label{sec:arguing_main_lemma}

For convenience, we recall the main result that is to be proven in this section:
\MainLemma*

In the following, we consider a potential future edge $e := \{u,v\} \in F_0$. Moreover, we fix its arrival time $T_e = t_e$ arbitrarily, and its \emph{potential neighborhood} $U_e$, which is defined as follows:
\begin{definition} \label{def:potential_neighborhood}
    For a fixed edge $e = \{u,v\}$, and a fixed arrival time $t_e$ of this edge, we call any set $U_e$ of the form $U_e = \{(f,t_f) : f \cap e \neq \emptyset, t_f < t_e\}$ a \emph{potential neighborhood} of the edge $e$. Here, each pair $(f,t_f)$ stands for an edge $f$ intersecting with $e$, and its arrival time $t_f \leq n^2$. Further, for $w \in \{u,v\}$, let $U_w = \{(f,t_f) \in U_e : w \in f \}$ be the corresponding induced neighborhood of vertex $w$. We only consider \emph{potential neighborhoods} $U_e$ such that $\max\{|U_u|, |U_v|\} \leq \Delta$, since the maximum degree of the graph is guaranteed to be $\Delta$.
\end{definition}
How many such pairs \((t_e, U_e)\) are possible? Note that there are at most \( n^2 \) possible values for the arrival time \( t_e \), and each edge has at most \( 2\Delta \) neighbors. For each neighbor, we must choose the vertex it connects to (with \( n \) options) and the time it arrives (with \( n^2 \) options). This results in at most \( n^{6\Delta} \) possible values for \( U_e \). 

Our goal is to show that, regardless of the specific choice of \( t_e \) and \( U_e \), the probability---over the randomness of the algorithm---that edge \( e \) has many \emph{bad} colors is very small. Specifically, we aim to bound this probability by roughly \( 1/n^{10\Delta} \). By a union bound, this ensures that, no matter how the adaptive adversary selects the future neighborhood of $e$, with high probability, the edge \( e \) will have only a few bad colors.

The proof will proceed in multiple steps. First, we focus on a fixed $e := \{u,v\} \in F_0$ with arrival time $t_e$, and a fixed potential neighborhood $U_e$ of this edge. We prove preliminary upper bounds on $P^{(t)}_{ec}$ for time steps $t < t_e$ in \cref{sec:arguing_main_lemma_part1}. These bounds are in terms of other random variables, which we extend to accommodate subsets of colors in \cref{sec:arguing_main_lemma_part2}. In \cref{sec:arguing_main_lemma_part3} we prove concentration inequalities on these latter random variables, which turn out to be supermartingales.

\subsubsection{Setup} \label{sec:arguing_main_lemma_part1}

In the following, consider an arbitrary color $c \in \Calg$. The next definition introduces some essential quantities that we will make use of in the proof of \cref{lemma:main_lemma}:
\begin{definition} \label{def:introducing_Q}
    For any $t < t_e$, and for the fixed color $c \in \Calg$, let:
    \begin{equation*}
        S^{(t)}_{ec} = \prod_{\substack{f \in U_e \\ t_f \leq t}} \left( 1 - P^{(t_f - 1)}_{fc} \right).
    \end{equation*}
    For $w\in \{u,v\}$, any $(f,t_f) \in U_w$, and for the fixed color $c \in \Calg$, let:
    \begin{equation*}
        R^{(t)}_{fc} = P^{(t)}_{fc} \cdot \prod_{\substack{g \in U_w \\ t_g \leq t}} \left( 1 - P^{(t_g-1)}_{gc} \right).
    \end{equation*}
    Finally, define $Q^{(t)}_{U_{w}c} = \sum_{f \in U_w} R^{(\min\{t,t_f - 1)\}}_{fc}$. In particular, note that $Q^{(t_e-1)}_{U_{w}c} = \sum_{f \in U_w} R^{(t_f - 1)}_{fc}$.
    (Note that the definitions of $S_{ec}$, $R_{fc}$, and $Q_{U_wc}$ depend on the edge $e$ and neighborhood $U_e$ we already fixed).
\end{definition}
\paragraph{Intuition Behind \cref{def:introducing_Q}} Consider an execution of \cref{alg:edge-coloring} which leaves color $c$ free by the time $t_e$ at which edge $e$ arrives, and assume that there would be no \color{blue}($\star$)\color{black}-condition in \cref{line:star} for the edge~$e$, i.e., $P_{ec}$ would always be allowed to scale up. Then we would have $P^{(t)}_{ec} = P^{(0)}_{ec}/S^{(t)}_{ec}$, and so $S$ is defined as the ``non-bounded'' scaling-factor of $P$ for executions in which $c$ is available for $e$ on arrival. 

\begin{figure}[ht!]
\centering
\includegraphics[width=0.5\textwidth]{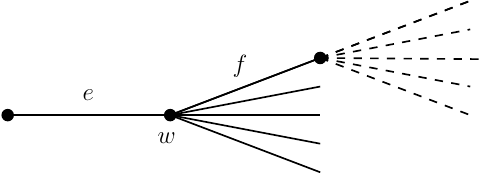}
\caption{The variables $R_{fc}$ do \emph{not} accumulate scalings induced by any of the solid edges, which arrived before $f$ and are incident at $w$.}
\label{fig:twohop}
\end{figure}

The variables $R^{(t)}_{fc}$, defined for potential future edges $f$ connected to $e$ at $w \in \{u,v\}$, accumulate the scalings produced by the arrival of other edges $g$ incident at $f$, but they ``revert'' all potential scalings caused by previous such edges which are also incident at $w$.\footnote{Note that even if, for some arriving $g \in U_w$ in \cref{alg:edge-coloring} with $t_g \leq t$,
\cref{alg:edge-coloring} did \emph{not} scale $P_{fc}$ up by $(1 - P^{(t_g-1)}_{gc})^{-1}$ (either because $Z_{g} > 1$ or $P_{fc} > A$), we still multiply the corresponding $R_{fc}$ by $(1 - P^{(t_g-1)}_{gc})$. In this case we ``revert'' the scaling up even though the scaling up is not performed. This is not essential in our definition, but it helps avoid having to introduce further notation.} For a better understanding, see \cref{fig:twohop}: Assume the edge $f \in U_w$ arrives second-to-last, before $e$, but after all the other edges in the picture. Any dashed edge connecting to $f$, which is in the $2$-hop of $w$, might lead to an increase of $R_{fc}$ due to scaling up. However, solid edges connecting to $w$, which also arrive before $f$, can never lead to an increase of $R_{fc}$ by definition, since we revert such a scaling if it happens.

The variables $Q^{(t)}_{U_wc}$ accumulate the sum of the $R^{(t)}_{fc}$-s (frozen at $t_f-1$ if $t > t_f-1$) for $f$ incident at $w$. They will be shown to be (almost) supermartingales, that can only increase due to arrivals of the $\leq \Delta^2$ many edges in the $2$-hop of $w$.
As we will see later in \cref{lemma:relating_S_to_Q}, $S_{ec} = (1-Q_{U_uc})\cdot(1-Q_{U_vc})$. Thus, showing bounds on the $Q$'s will be helpful to bound the scaling factor $S_{ec}$, and in turn show that the event $P_{ec} > A$ is unlikely (and only happens for few colors). Concretely, \cref{corollary:characterization_of_bad_colors} will reduce the task of showing that a color $c \in \Calg$ is not \emph{bad} for edge $e$ to the task of upper bounding $Q_{U_uc}$ and $Q_{U_vc}$:
\begin{lemma} \label{corollary:characterization_of_bad_colors}
    Consider instances in which edge $e$ arrives at time $t_e$ and has neighborhood given by $U_e$. If $Q^{(t_e-1)}_{U_wc} \leq 1 - \frac{\varepsilon}{2}$ for all $w \in \{u,v\}$, then the color $c \in \Calg$ is not bad for $e$, i.e., $c \notin B^{(t)}_{e}$ for all $t < t_e$. 
\end{lemma}
For any fixed color $c \in \Calg$, the initial value $Q^{(0)}_{U_wc}$ of $Q_{U_wc}$ is at most $1-\varepsilon$. It is not hard to show that $Q^{(t)}_{U_wc}$ increases by less than $\varepsilon/2$ with high probability in $n$. However, this is not enough, since we need to union bound over $n^{\Theta(\Delta)}$ events (for each possible fixed $t_e$ and fixed potential neighborhood $U_e$). If the events $Q^{(t)}_{U_wc} > Q^{(0)}_{U_wc} + \varepsilon/2$ were independent across colors $c \in \Calg$, we would however be able to conclude that any $\sim \varepsilon^5 \Delta$ colors can be simultaneously bad with at most $1/n^{\Theta(\varepsilon^5 \Delta)}$ probability, which would allow applying the previously mentioned union bound. Because these events are in fact correlated, we will prove this fact by the use of martingales in \cref{sec:arguing_main_lemma_part2,sec:arguing_main_lemma_part3}.

We resume by providing the proof of \cref{corollary:characterization_of_bad_colors}, which is obtained from the following sequence of lemmas:
\begin{lemma} \label{lemma:ub_on_P_in_terms_of_S}
    Consider instances in which edge $e$ arrives at time $t_e$ and has neighborhood given by $U_e$. For any $t < t_e$, we have:
    \begin{equation*}
        P^{(t)}_{ec} \leq \frac{1 - \varepsilon}{\Delta} \cdot \frac{1}{S^{(t)}_{ec}}.
    \end{equation*}
\end{lemma}
\begin{proof}
    We prove the inequality by induction over $t \geq 0$. For $t = 0$ we have $P^{(0)}_{ec} = (1 - \varepsilon)/\Delta$ and $S^{(0)}_{ec} = 1$, such that the inequality is clearly true. Now consider a time step $(t + 1) < T_e$, and assume $P^{(t)}_{ec} \leq (1-\varepsilon)/\Delta \cdot 1/(S^{(t)}_{ec})$. By the definition of $S^{(t)}_{ec}$, it is clear that $S^{(t)}_{ec} \geq S^{(t+1)}_{ec}$. Hence, if $P^{(t+1)}_{ec} = P^{(t)}_{ec}$, it immediately follows that $P^{(t+1)}_{ec} \leq (1 - \varepsilon)/\Delta \cdot 1/(S^{(t+1)}_{ec})$. Hence, we may assume that $(t+1)$ is a time step such that $P^{(t+1)}_{ec} \neq P^{(t)}_{ec}$. 
    
    Further, if $P^{(t+1)}_{ec} = 0$ the inequality to be proven follows immediately, so we consider the remaining case, which is that $P^{(t+1)}_{ec}$ is scaled up at this time step. This implies that the arriving edge $e_{t+1}$ at time $t$ is such that $e_{t+1} \in U_e$ with arriving time $T_{e_{t+1}} = t+1$. Hence:
    \begin{equation*}
        P^{(t+1)}_{ec} = \frac{P^{(t)}_{ec}}{1 - P^{(t)}_{e_{t+1}c}} \ \ \text{ and  } \ \ 
        S^{(t+1)}_{ec} = S^{(t)}_{ec} \cdot \left( 1 - P^{(t)}_{e_{t+1}c}\right).
    \end{equation*}
    These two identities, together with the assumption by induction that $P^{(t)}_{ec} \leq (1 - \varepsilon)/\Delta \cdot 1/(S^{(t)}_{ec})$, immediately imply the desired inequality that $P^{(t+1)}_{ec} \leq (1 - \varepsilon)/\Delta \cdot 1/(S^{(t+1)}_{ec})$.
\end{proof}
Now we show that the scaling factor $S_{ec}$ for edge $e = \{u,v\}$ can be decomposed into scaling factors for each vertex $w\in \{u,v\}$.
\begin{lemma} \label{lemma:relating_S_to_Q}
    Consider instances in which edge $e$ arrives at time $t_e$ and has neighborhood given by $U_e$. We have:
    \begin{equation*}
        S^{(t_e-1)}_{ec} = \left(1 - Q^{(t_e-1)}_{U_uc} \right) \cdot \left(1 - Q^{(t_e-1)}_{U_vc} \right).
    \end{equation*}
\end{lemma}
\begin{proof}
    Since $U_e = U_u \cup U_v$ and $U_u \cap U_v = \emptyset$, it suffices to prove that for any fixed $w \in \{u,v\}$:
    \begin{equation} \label{eq:relating_S_to_Q_splitting_prod}
        \prod_{\substack{f \in U_w \\ t_f \leq t}} \left( 1 - P^{(t_f-1)}_{fc} \right) = 1 - Q^{(t_e-1)}_{U_wc}.
    \end{equation}
    Order all edges $f_1, f_2,\dots,f_k \in U_w$ by their arrival times $t_{f_1}<t_{f_2}<\dots,t_{f_k} < t_e$. We prove by induction on $0 \leq j \leq k$ that:
    \begin{equation} \label{eq:induction_for_relating_S_to_Q}
        \prod_{i = 1}^j \left( 1 - P^{(t_{f_i}-1)}_{f_ic} \right) = 1 -\sum_{i=1}^j R^{(t_{f_i}-1)}_{f_ic}.
    \end{equation}
    Note that for $j = k$ this implies \eqref{eq:relating_S_to_Q_splitting_prod} and thus the statement of the lemma. For $j = 0$ the equality is clear, such that we proceed with the induction step $j \to j+1$. We use \cref{def:introducing_Q} and the fact that $\{g \in U_w : t_g < t_{f_{j+1}}\} = \{f_1,\dots,f_j\}$ to obtain that:
    \begin{align*}
        \prod_{i = 1}^{j+1} \left( 1 - P^{(t_{f_i}-1)}_{f_ic} \right) &= \prod_{i = 1}^{j} \left( 1 - P^{(t_{f_i}-1)}_{f_ic} \right) \cdot \left( 1 - \frac{R^{(t_{f_{j+1}}-1)}_{f_{j+1}c}}{\prod_{\substack{g \in U_w \\ T_g < T_{f_{j+1}} }} \left( 1 - P^{(t_{g}-1)}_{gc} \right)} \right)\\
        &=\prod_{i = 1}^{j} \left( 1 - P^{(t_{f_i}-1)}_{f_ic} \right) \cdot \left( 1 - \frac{R^{(t_{f_{j+1}}-1)}_{f_{j+1}c}}{\prod_{i=1}^j \left( 1 - P^{(t_{f_i}-1)}_{f_ic} \right)} \right) \\
        &= \prod_{i = 1}^{j} \left( 1 - P^{(t_{f_i}-1)}_{f_ic} \right) - R^{(t_{f_{j+1}}-1)}_{f_{j+1}c}.
    \end{align*}
    Applying the induction hypothesis for the product on the right hand side gives \eqref{eq:relating_S_to_Q_splitting_prod} as desired, and the statement of the lemma follows.
\end{proof}

\begin{proof}[Proof of \cref{corollary:characterization_of_bad_colors}.]
    From \cref{lemma:ub_on_P_in_terms_of_S} and \cref{lemma:relating_S_to_Q}, we have:
    \begin{equation*}
        P^{(t)}_{ec} \leq \frac{1 - \varepsilon}{\Delta} \cdot \frac{1}{\left(1 - Q^{(t_e-1)}_{U_uc}\right) \left(1 - Q^{(t_e-1)}_{U_vc}\right)}.
    \end{equation*}
    Under the assumed upper bound on $Q^{(t_e-1)}_{U_wc}$, for both $w \in \{u,v\}$, this can in turn be upper bounded by:
    \begin{equation*}
        \frac{1 - \varepsilon}{\Delta} \cdot \frac{1}{\left( \frac{\varepsilon}{2}\right)^2} = \frac{4 \cdot (1 - \varepsilon)}{\varepsilon^2 \Delta} \leq \frac{4}{\varepsilon^2 \Delta} \leq A.
    \end{equation*}
    This shows that $P^{(t)}_{ec} \leq A$ and thus $c \notin B^{(t)}_e$.
\end{proof}

\paragraph{Preliminary Conclusions.} Recall that we aim to prove \cref{lemma:main_lemma}, i.e., to prove a (high probability) $2\varepsilon^5 \Delta$ upper bound on the cardinality of the set of bad colors $B^{(t)}_{e} = \{c \in \Calg : P^{(t)}_{ec} > A\}$. \cref{corollary:characterization_of_bad_colors} reduces the task of proving a color is \emph{not} bad to the task of upper bounding $Q^{(t_e-1)}_{U_wc}$ for both vertices $w \in \{u,v\}$. In the next section we proceed to deal with this task.

\subsubsection{Upper Bound on the Scaling}  \label{sec:arguing_main_lemma_part2}

To summarize the work done so far, we fixed a potential future edge $e := \{u,v\}$, its arrival time $t_e$, and its \emph{potential neighborhood} $U_e$ (as defined in \cref{def:potential_neighborhood}). Then, we showed in \cref{corollary:characterization_of_bad_colors} that a color $c \in \Calg$ cannot be bad if $Q^{(t_e-1)}_{U_wc}$ is bounded from above for all $w \in \{u,v\}$. Proving \cref{lemma:main_lemma} requires showing that $|B^{(t)}_{f}| \leq 2 \varepsilon^5 \Delta$ with high probability for any time step $t < t_e$, where $B^{(t)}_{e} = \{c \in \Calg : P^{(t)}_{ec} > A\}$.

For the rest of this section, we additionally fix a particular $w \in \{u,v\}$. Further, we concentrate on a fixed subset $C \subseteq \Calg$ of $2 \varepsilon^5 \Delta$ colors, i.e., $|C| = 2 \varepsilon^5 \Delta$, and give an upper bound on the probability that \emph{all} of these colors are bad at some time step $t$. We extend the definition of $Q$ in the following way to account for the fixed subset $C$ of colors:
\begin{definition} \label{def:introducing_Q_extended_extended}
    For any $t < t_e$, $w \in \{u,v\}$ and for any subset of colors $C \subseteq \Calg$, let:
    \begin{equation*}
        Q^{(t)}_{U_wC} = \sum_{c \in C} Q^{(t)}_{U_wc}.
    \end{equation*}
\end{definition}
We will prove the following:
\begin{restatable}{lemma}{LemmaBoundsOnQ} \label{lemma:bounds_on_Q}
    Fix $w \in \{u,v\}$, and fix a subset of $ \varepsilon^5 \Delta$ colors $C \subseteq \Calg$. Then, we have:
    \begin{equation*}
        \Pr\left[ Q^{(T_e-1)}_{U_wC} \geq \varepsilon^5\Delta - \frac{\varepsilon^6\Delta}{2} \right] \leq \frac{1}{n^{10\Delta}}.
    \end{equation*}
\end{restatable}
To avoid disrupting the flow of the argument we provide the proof of this lemma in \cref{sec:arguing_main_lemma_part3}. In the following we show how this last lemma allows to prove \cref{lemma:main_lemma}, by formalizing the union bound to which we alluded in \cref{sec:arguing_main_lemma_part1}:
\MainLemma*
\begin{proof}[Proof of \Cref{lemma:main_lemma}]
    We prove that, for any time step $t$:
    \begin{equation*}
        \Pr\left[ \left|B^{(t)}_{e}\right| > 2\varepsilon^5 \Delta \right] \leq \frac{1}{n^{100}}.
    \end{equation*}
    Let $\binom{\Calg}{2 \varepsilon^5 \Delta}$ be the set of all subsets of $2 \varepsilon^5 \Delta$ colors from $\Calg$. Note that:
    \begin{align*}
        \Pr\left[ \left|B^{(t)}_{e}\right| > 2 \varepsilon^5 \Delta \right] \leq \Pr\left[\exists C \in \binom{\Calg}{2 \varepsilon^5 \Delta} : C \subseteq B^{(t)}_e\right] 
    \end{align*}
    For $w \in \{u,v\}$, let $B^{(t)}_e(w) = \{e \in B^{(t)}_e : Q^{(t_e-1)}_{wc} > 1 - \frac{\varepsilon}{2}\}$. From \cref{corollary:characterization_of_bad_colors} it follows that $B^{(t)}_e \subseteq B^{(t)}_e(u) \cup B^{(t)}_e(v)$. We have that:
    \begin{align*}
        \Pr\left[\exists C \in \binom{\Calg}{2 \varepsilon^5 \Delta} : C \subseteq B^{(t)}_e\right] &\leq \Pr\left[\exists C \in \binom{\Calg}{\varepsilon^5\Delta} \ \exists w \in \{u,v\} \ \text{s.t.} \ C \subseteq B^{(t)}_e(w)\right] \\
        &\leq \Pr\left[\exists C \in \binom{\Calg}{\varepsilon^5\Delta} \ \exists w \in \{u,v\} \ \forall c \in C : Q^{(t_e-1)}_{wc} > 1 - \frac{\varepsilon}{2}\right] \\
        &\leq \Pr\left[\exists C \in \binom{\Calg}{\varepsilon^5\Delta} \ \exists w \in \{u,v\} \ \text{s.t.}\  \sum_{c \in C}Q^{(t_e-1)}_{wc} > \varepsilon^5\Delta - \frac{\varepsilon^6\Delta}{2}\right] \\
        &\leq \sum_{C \in \binom{\Calg}{\varepsilon^5 \Delta}}\sum_{w \in \{u,v\}} \Pr\left[ \sum_{c \in C}Q^{(t_e-1)}_{wc} > \varepsilon^5\Delta - \frac{\varepsilon^6\Delta}{2}\right].
    \end{align*}
    Let $\mathcal{U}_e = \{ \{(f,t_f) : f \cap e \neq \emptyset, t_f < t_e\} \}$ be the set of all possible potential neighborhoods $U_e$ of $e$. Recall that any $U_e \in \mathcal{U}_e$ induces a corresponding partition $U_e = U_u \cup U_v$ (we avoid heavier notation such as $U_u(e)$ and $U_v(e)$ for these subsets, but insist on the fact they depend on $e$ and are determined by $U_e$).

    Since, for any $w \in \{u,v\}$, $\sum_{c \in C} Q^{(t_e-1)}_{wc} \in \left\{Q^{(t_e-1)}_{U_wC} : U_e \in \mathcal{U}_e\right\}$, we have for any such $w$:
    \begin{align*}
        \Pr\left[ \sum_{c \in C}Q^{(t_e-1)}_{wc} > \varepsilon^5\Delta - \frac{\varepsilon^6\Delta}{2}\right] &\leq \Pr\left[ \exists U_e \in \mathcal{U}_e \ \text{s.t.} \ Q^{(t_e-1)}_{U_wC} > \varepsilon^5\Delta - \frac{\varepsilon^6\Delta}{2} \right] \\
        &\leq \sum_{U_e \in \mathcal{U}_e} \Pr\left[Q^{(t_e-1)}_{U_wC} > \varepsilon^5\Delta - \frac{\varepsilon^6\Delta}{2} \right].
    \end{align*}
    Combining all previous inequalities, we get:
    \begin{equation*}
        \Pr\left[ \left|B^{(t)}_{e}\right| > 2 \varepsilon^5 \Delta \right] \leq \sum_{C \in \binom{\Calg}{\varepsilon^5 \Delta}} \sum_{w \in \{u,v\}} \sum_{U_e \in \mathcal{U}_e} \Pr\left[Q^{(t_e-1)}_{U_wC} > \varepsilon^5\Delta - \frac{\varepsilon^6\Delta}{2} \right].
    \end{equation*}
    Upper bounding the relevant cardinalities:
    \begin{equation*}
        \left| \binom{\Calg}{\varepsilon^5 \Delta} \right| \cdot \left|\{u,v\}\right| \cdot \left| \mathcal{U}_e \right| \leq 2^{\Delta} \cdot 2 \cdot n^{7\Delta} \leq n^\Delta \cdot n^\Delta \cdot n^{7\Delta} = n^{9\Delta},
    \end{equation*}
    (recall that $|\mathcal{U}_e| \leq n^{7 \Delta}$ by the discussion after \cref{def:potential_neighborhood}) and using \cref{lemma:bounds_on_Q}, which holds for any $C$, $w$ and $U_e$, we get the desired bound:
    \begin{align*}
        \Pr\left[ \left|B^{(t)}_{e}\right| > 2 \varepsilon^5 \Delta \right] & \leq n^{9\Delta} \cdot \frac{1}{n^{10\Delta}} = \frac{1}{n^\Delta} \leq \frac{1}{n^{100}}. \qedhere
    \end{align*}
\end{proof}

\subsubsection{Proving Bounds on $Q$}  \label{sec:arguing_main_lemma_part3}

In this section we prove the required bounds on $Q$.
\LemmaBoundsOnQ*
This will follow quite easily by first noting that $Q$ behaves as a supermartingale (with some additional negative drift) over time, and then by applying Azuma's inequality:
\begin{lemma} \label{lemma:Q_is_supermartingale}
    The sequence of random variables $Q^{(0)}_{U_wC}, Q^{(1)}_{U_wC}, \dots$ is upper bounded by a supermartingale with respect to the sequence of random variables $K_1, K_2, \dots$. Furthermore, the step size of this supermartingale is bounded by $12A$. 
\end{lemma}
\begin{proof}
    For convenience, we recall the definition of $Q^{(t)}_{U_wC}$; we have:
    \begin{equation*}
        Q^{(t)}_{U_wC} = \sum_{c \in C} \sum_{(f,T_f) \in U_w} R^{(\min\{t, T_f-1\})}_{fc},
    \end{equation*}
    where the variables $R_{fc}$ are defined in \cref{def:introducing_Q}. 
    
    Fix $t \geq 0$ and consider the $(t+1)$-th time step, where edge $g := e_{t+1}$ arrives. We may assume $Z^{(t)}_{e_{t+1}} \leq 1$, since otherwise the edge is \emph{marked} in \cref{line:mark_z_ge_1} and $Q$ does not change. Assuming this, note that the only $R_{fc}$'s that change in this time step are those with edges $f\in U_w$ that share a vertex with edge $g$. In case $g$ is not incident to $w$, then let $G_{t+1} \subseteq U_w$ be the set of the (up to two) edges $f \in U_w$ incident to $g$. Otherwise, if $g$ is incident to $w$, we note that for edges $f\in U_w$ the values $R_{fc}$ will not increase by definition of $R_{fc}$, i.e., $R^{(t+1)}_{fc} \leq R^{(t)}_{fc}$, so in this case we will let $G_{t+1}=\emptyset$.
    Define:
    \begin{equation*}
        (\delta Q)^{(t + 1)}_{U_wC} = \sum_{c \in C} \sum_{f \in G_{t+1}} \left( R^{(t+1)}_{fc} - R^{(t)}_{fc} \right).
    \end{equation*}
    This is the change in $Q$ limited to the non-trivial edges in $G_{t+1}$. By the above discussion, it follows that:
    \begin{equation*}
        Q^{(t + 1)}_{U_wC} - Q^{(t)}_{U_wC} \leq (\delta Q)^{(t + 1)}_{U_wC}.
    \end{equation*}    
        
    We next prove that the sum $\sum_{t'} (\delta Q)^{(t')}_{U_wC}$ behaves as as supermartingale with bounded step size. Again consider the fixed $t \geq 0$ and the $(t+1)$-th time step. (Again, we may assume $Z^{(t)}_{e_{t+1}} \leq 1$, since otherwise the edge is \emph{marked} in \cref{line:mark_z_ge_1} and $Q$, $\delta Q$ both do not change). For each edge $f \in G_{t+1}$, we consider its set of bad edges $B^{(t)}_f$ (we will not use any non-trivial facts about this set). Further, let $K_{t+1} \in \Calg \cup \{\perp\}$ be the color sampled by \cref{alg:edge-coloring} in \cref{line:sample}; we have:
    \begin{align*}
        &\E\left[(\delta Q)^{(t + 1)}_{U_wC} \;\middle\vert\; K_1,\dots,K_t \right] \leq \sum_{f \in G_{t+1}} \sum_{c \in C} \E\left[ \left( R^{(t+1)}_{fc} - R^{(t)}_{fc}\right) \;\middle\vert\; K_1,\dots,K_t \right]\\
        &\leq \sum_{f \in G_{t+1}} \sum_{c \in C} \left[ \Pr\left[K_{t+1} = c\mid K_{\le t} \right] \cdot \left( -R^{(t)}_{fc} \right) + \mathds{1}[c \notin B^{(t)}_f] \cdot \left( 1 - \Pr\left[K_{t+1} = c \mid K_{\le t} \right] \right) \cdot \left( \frac{R^{(t)}_{fc}}{1 - P^{(t)}_{e_{t+1}c}} - R^{(t)}_{fc} \right) \right] \\
        &\leq \sum_{f \in G_{t+1} } \sum_{c \in C} \left[ P^{(t)}_{e_{t+1}c} \cdot \left( -R^{(t)}_{fc} \right) + \left( 1 - P^{(t)}_{e_{t+1}c} \right) \cdot \left( \frac{R^{(t)}_{fc}}{1 - P^{(t)}_{e_{t+1}c}} - R^{(t)}_{fc} \right) \right] \\
        &= 0.
    \end{align*}
    This proves that $\sum_{t'} (\delta Q)^{(t')}_{U_wC}$ is a supermartingale. It remains to upper bound its step size:
    \begin{align*}
        \left| (\delta Q)^{(t + 1)}_{U_wC} \right| &
        \leq \sum_{f \in G_t} \left|-\mathds{1}[K_{t+1} \in C] \cdot R^{(t)}_{fK_{t+1}} + \sum_{c \in C \setminus (B^{(t)}_f \cup \{K_{t+1}\})} \left( \frac{R^{(t)}_{fc}}{1 - P^{(t)}_{e_{t+1}c}} - R^{(t)}_{fc}\right)\right| \\
        &\leq \sum_{f \in G_t}\mathds{1}[K_{t+1} \neq \perp] \cdot R^{(t)}_{fK_{t+1}} + \sum_{f \in F}\sum_{c \in \Calg} \left( \frac{R^{(t)}_{fc}}{1 - P^{(t)}_{e_{t+1}c}} - R^{(t)}_{fc}\right) \\
        &\leq \sum_{f \in G_t}2A + \sum_{f \in F}\sum_{c \in \Calg} \frac{P^{(t)}_{fc} \cdot P^{(t)}_{e_{t+1}c}}{1 - P^{(t)}_{e_{t+1}c}} \\
        &\leq 4A + 4\cdot \sum_{c \in \Calg} P^{(t)}_{fc} \cdot P^{(t)}_{e_{t+1}c}.
    \end{align*}
    Above, we first make use of the triangle inequality, and increase the second sum by adding (non-negative) terms for all colors in $\Calg$. Then, we use the invariant from \cref{lemma:invariant_ub_on_P} and the fact that $R \leq P$ always (as a trivial consequence of \cref{def:introducing_Q}) to upper bound the first term by $2A$ and replace the $R$-s by $P$-s in the second term, and then use the fact fact that $|G_t| \leq 2$ and the invariant again to argue $1/(1 - P^{(t)}_{e_{t+1}c}) \leq 1/(1-2A) \leq 1/2$. This is very similar to the proof of \cref{lemma:Z_is_supermartingale}. We refer to that proof to argue that $\sum_{c \in \Calg} P^{(t)}_{fc} \cdot P^{(t)}_{e_{t+1}c} \leq 2A$, and obtain the desired upper bound of $12A$ for the step size.
\end{proof}

\begin{proof}[Proof of \cref{lemma:bounds_on_Q}]
    We have that, for any $(f,t_f) \in U_w$ and $c \in C$, $R^{(0)}_{fc} = \frac{1 - \varepsilon}{\Delta}$. This implies that:
    \begin{equation*}
        Q^{(0)}_{U_wC} = \sum_{c \in C} \sum_{(f,t_f) \in U_w} \frac{1 - \varepsilon}{\Delta} = |C| \cdot |U_w| \cdot \frac{1 - \varepsilon}{\Delta} \leq \varepsilon^5 \Delta \cdot \Delta \cdot \frac{1 - \varepsilon}{\Delta} = \varepsilon^5\Delta - \varepsilon^6 \Delta.
    \end{equation*}
    The inequality follows from the fact that $|U_w| \leq \Delta$, which holds because by assumption $\Delta$ is the maximum degree of the input graph, and thus we only consider potential neighborhoods of size at most $\Delta$ for the vertex $w$. The above upper bound on the initial value of the supermartingale $Q_{U_wC}$\footnote{While not strictly a supermartingale, the previous lemma allows us to treat it as such.} allows to rewrite events as follows:
    \begin{equation*}
        \Pr\left[ Q^{(t_e-1)}_{U_wC} \geq \varepsilon^5\Delta - \frac{\varepsilon^6\Delta}{2} \right] \leq \Pr\left[ Q^{(t_e-1)}_{U_wC} - Q^{(0)}_{U_wC} \geq \frac{\varepsilon^6\Delta}{2} \right].
    \end{equation*}
    To set up Azuma's inequality, let $\lambda = (\varepsilon^6 \Delta) / 2$ and let $T^{+}_e$ be $\left|\left\{t < t_e - 1: Q^{(t+1)}_{U_wC} \neq Q^{(t)}_{U_wC}\right\}\right|$. These are the number of time steps that lead to (non-zero) changes in the supermartingale $Q_{U_wC}$. Apply \cref{lemma:azuma} to obtain:
    \begin{align*}
        \Pr\left[ Q^{(t_e-1)}_{U_wC} - Q^{(0)}_{U_wC} \geq \frac{\varepsilon^6\Delta}{2} \right] \leq \exp\left( - \frac{\lambda^2}{2 (T^+_e-1) (12A)^2} \right).
    \end{align*}
    It suffices to argue that:
    \begin{equation*}
        \frac{\lambda^2}{2 (T^+_e-1) (12A)^2} \geq 10 \Delta \ln n.
    \end{equation*}
    We observe that $T^+_e \leq 2\Delta^2$. This is the case because, by \cref{def:introducing_Q}, $Q_{U_wC}$ can only change due to online arrivals of edges $g$ such that $g \cap f \neq \emptyset$ for some potential future arrival $(f,t_f) \in U_w$. Since $\Delta$ is the maximum degree in the graph, there are at most $\Delta + \Delta^2 \leq 2\Delta^2$ such edges.\footnote{In fact, by a slightly more careful inspection, one can see that $Q_{U_wC}$ only changes due to arriving edges in the (strict) $2$-hop neighborhood (w.r.t. $U_w$) of $w$, which can give a more precise upper bound of $\Delta^2$ edges.} Hence, it suffices to argue that:
    \begin{equation*}
        \frac{\lambda^2}{576 \cdot \Delta^2 A^2} \geq 10 \Delta \ln n.
    \end{equation*}
    Plugging in $\lambda = (\varepsilon^6 \Delta) / 2$ and $A = c_A/(\varepsilon^2 \Delta)$, the above inequality can be simplified to:
    \begin{equation*}
        \frac{1}{(c_A \cdot 48)^2} \cdot \varepsilon^{16}\Delta^2 \geq 10\Delta \ln n.
    \end{equation*}
    Recalling $\varepsilon = c_\varepsilon \cdot (\ln n / \Delta)^{1/16}$, this is equivalent to:
    \begin{equation*}
        \frac{1}{(c_A \cdot 48)^2} \cdot c^{16}_{\varepsilon} \cdot \Delta \ln n \geq 10\Delta \ln n,
    \end{equation*}
    which holds by the choices of the constants $c_A$ and $c_\varepsilon$.
\end{proof}

\section{Randomized (Oblivious) Online Edge Coloring}
\label{sec:oblivious}

In this section we prove that when $\Delta = \omega(\sqrt{\log n})$ one can edge-color a graph online using at most $(1+o(1))\Delta$  colors with a \emph{randomized} algorithm against an \emph{oblivious} adversary.

\begin{theorem}
\label{thm:oblivious_theorem}
There is a randomized online algorithm that edge-colors $n$-node graphs with known $n$ and $\Delta$, generated obliviously, using with high probability
$$
\Delta + O(\Delta^{15/16}\log^{1/32}n) \textrm{ colors}.
$$
\end{theorem}

\paragraph{Challenges in Adapting \cref{alg:edge-coloring}.}
The algorithm is very similar to the one in \cref{sec:adaptive} (adaptive adversaries in the regime when $\Delta = \omega(\log n)$). The $\exp(-\poly(\eps)\Delta^2)$ concentration in \cref{sec:adaptive} was used to bound over $n^{\Theta(\Delta)}$ possible futures in the case of an adaptive adversary. In this section we show that similar $\exp(-\poly(\eps)\Delta^2)$ concentration can be used already when $\Delta \approx \sqrt{\log n}$ to union bound over $\poly(n)$ events when we have the weaker oblivious adversary.

Our algorithm must be modified due to the following fact: if the graph consists of $\approx n/\Delta$ stars, where each edge is colored with probability $1-\eps$ using $\Calg$ (the good case for the algorithm in \cref{sec:adaptive}), then one would unfortunately expect one of these stars to have $\approx \eps \Delta + \log n$ many $\Calg$-uncolored (i.e.~\emph{marked}) edges, but this is too much when $\Delta \approx \sqrt{\log n}$.

Fortunately, only few vertices will have so many $\Calg$-uncolored edges. In fact, it is not difficult to show that 
only a $\approx \frac{1}{2^{\Delta}}$-fraction of the vertices have more than $1000\eps \Delta$ such $\Calg$-uncolored edges, and these \emph{bad} vertices are reasonably spread out in the graph. 
In particular, we can prove that any vertex $v$ neighbors at most $\eps\Delta$ many bad vertices, with probability $1-\exp(-\poly(\eps)\Delta^2) \gg 1-\frac{1}{n^{100}}$. This would be true by standard Chernoff bounds if each vertex was independently bad with probability $\frac{1}{2^{\Delta}}$; as these events are \emph{not} independent we again rely on martingales to handle the correlations and prove this statement in our analysis.

Thus, we will modify the algorithm to handle these \emph{bad} vertices by insisting that their remaining edges must be colored using a color from $\Calg$ (if possible). Since these bad vertices are few and reasonably spread out in the graph, this modification does not affect the rest of the algorithm and analysis too much.

\paragraph{Stars vs Matchings.} Several of our proofs require arguing $\exp(-\poly(\eps)\Delta^2)$ concentration on certain events relating to (fixed) \emph{subsets} $M$ of $\poly(\Delta)$ edges. Indeed, when $\Delta \approx \sqrt{\log n}$, this would imply a desired ``with high probability in $n$'' concentration. In \cref{sec:adaptive}, martingale arguments allowed us to obtain such guarantees without needing to argue about the correlations explicitly, and our approach will be similar here. However, there is a technical difficulty which requires further care. Intuitively, the problem arises when trying to compute the step size of a supermartingale of the form $\sum_{e \in M} Z_e$. Depending on the structure of $M$, the step size might be large. Indeed, if the \emph{subset} $M$ includes a star of $\Delta$ edges, then any such arriving edge affects $\Omega(\Delta)$ other edges, resulting in a step size $\Omega(\Delta)$ times larger than ``usual''. This is too large for our purposes, and to circumvent the issue we will decompose $M$ into \emph{matchings}. If $M$ is a matching, an arriving edge $f$ can touch at most two edges in $M$, which allows to control the step size. 

More concretely, our proofs will require arguing that for any matching $M$, and any (reasonably large) subset of colors $C$, the quantity $\sum_{e \in M} \sum_{c \in C} P_{ec}$ behaves as a well concentrated martingale. This is formalized in \cref{lemma:matching_lemma}. The discussion preceding the lemma gives more details on our approach.

\subsection{Notation and Parameters} \label{sec:notation_assumptions_oblivious}

In \cref{sec:algorithm_oblivious} we introduce our algorithm (\cref{alg:edge-coloring-oblivious}), which is a modification of \cref{alg:edge-coloring} from \cref{sec:algorithm}. In general, we will make use of notation introduced previously in \cref{sec:algorithm} and \cref{sec:analysis}.
Throughout the rest of the paper we assume an oblivious adversary.

\paragraph{Redefining the Parameter $\varepsilon$:} We redefine $\varepsilon$ to be $\varepsilon := c_\varepsilon \cdot \left( \frac{\sqrt{\ln n}}{\Delta} \right)^{1/16}$ (note the $\sqrt{\ln}$ instead of $\ln n$ as in \cref{sec:adaptive}). Intuitively, this is necessary because we aim to use roughly $\varepsilon \Delta$ extra colors, and this quantity should be $o(\Delta)$ for $\Delta = \omega(\sqrt{\log n})$.
The other parameters and constants ($A,c_\eps, c_A$) from \cref{sec:adaptive} remain as before. For our analysis of the new algorithm, we also need a new parameter $\alpha$ and constant $c_K$ relating to \emph{bad} vertices. 

\begin{definition}[Parameters for Oblivious Case] \label{def:parameters-obliv}
    Define $\varepsilon := c_\varepsilon \cdot \left( \frac{\sqrt{\ln n}}{\Delta} \right)^{1/16}$ and $A := \frac{c_A}{\varepsilon^2 \Delta}$. Here, $c_\varepsilon := 10$ and $c_A := 4$ are the same constants as before. Additionally, let $c_K := 35c^2_A$ be a new constant,
    and let $\alpha := \varepsilon^3 / 100$. 
\end{definition}

Similar to \cref{sec:adaptive}, our target will be to show that the algorithm uses   $(1+O(\eps)) \Delta$ colors, but now with a smaller $\eps$ giving a total of $\Delta + O(\Delta^{15/16}\log^{1/32} n)$ colors.
    When $\Delta < (10c_{\eps})^{16} \sqrt{\ln n}$,
this is trivial, as the algorithm never uses more than $3\Delta$ colors. In the rest of this section we thus assume, without loss of generality, that $\Delta \ge (10c_{\eps})^{16} \sqrt{\ln n}$, implying that $\eps$ and $A$ are both small.

\begin{fact}
\label{fact:delta_large_enough-obliv}
    When $\Delta \ge (10c_{\eps})^{16} \sqrt{\ln n}$, then $\eps \le \frac{1}{10}$ and $A \le \frac{1}{10}$.
\end{fact}

\paragraph{Bad Vertices.}
For any vertex $v \in V$, \cref{alg:edge-coloring-oblivious} introduces a new corresponding variable $\udeg(v)$, intuitively keeping track of the number of incident edges colored by the backup palette $\Cgreedy$. When the $\udeg$ of a vertex is too large, we call it $\bad$, and change our logic to handle future incident edges to it. We discuss this quantity in more detail after introducing the algorithm. Below, we define what it means for a vertex to be \emph{bad} and \emph{dangerous}. In the analysis we show that, with high probability, few vertices are bad (and they are spread out in the graph), while no vertices are dangerous.
\begin{definition}[Bad and Dangerous Vertices] \label{def:bad_vertices_etc}
    Consider a fixed instance of the online edge coloring problem, and let $v \in V$ be an arbitrary vertex. Let $\udeg^{(t)}(v)$ be the value of $\udeg(v)$ at time step $t$, during the execution of \cref{alg:edge-coloring-oblivious}. Vertex $v$ is \emph{bad} (at time step $t$), if $\udeg^{(t)}(v) \geq 2c_K \cdot  \varepsilon \Delta$. Else, $v$ is \emph{good}.

    Further, for any vertex $v \in V$ and time step $t$, let $\baddeg^{(t)}(v)$ be the number of neighbors $w$ of $v$ such that: the edge $f := \{v,w\}$ arrived at time $t_f$ and vertex $w$ was \emph{bad} at time step $t_f$. This is the number of neighbors of $v$, taken up to time step $t$, which were \emph{bad} when they were connected to $v$. 
    Vertex $v$ is \emph{dangerous} (at time $t$) if $\baddeg^{(t)}(v) \geq \alpha \Delta$.
\end{definition}

\subsection{The Algorithm} \label{sec:algorithm_oblivious}

\begin{algorithm}[ht!]
	\caption{Edge Coloring Algorithm Against {\color{green!50!black}Oblivious} Adversaries}
	\label{alg:edge-coloring-oblivious}
	\begin{algorithmic}[1]
        \Statex \underline{\smash{\textbf{Input:}}} Vertex set $V$ and maximum degree $\Delta\in \mathbb{Z}_{\geq 0}$ of the graph to arrive
        \Statex 
                      \textbf{\underline{Initialization:}} 
        $\C_{\mathrm{alg}} \leftarrow [\Delta]$, 
        $F_0 \leftarrow \binom{V}{2}$,
        {\color{green!50!black}$\eps \gets c_{\eps}\cdot(\frac{\sqrt{\ln n}}{\Delta})^{1/16}$},
        $A \gets \frac{c_{A}}{\eps^2\Delta}$, where $c_\varepsilon := 10$ and $c_A := 4$.
        \Statex\qquad\qquad\qquad\;\;\textbf{for each} $e \in F_0$ and $c \in \Calg$:\; Set $P^{(0)}_{ec} \leftarrow \frac{1 - \varepsilon}{\Delta}$.
        
        \Statex 
		\For{\textbf{each} online edge $e_t=\{u,v\}$ on arrival} 
            \State $F_t \leftarrow F_{t-1} \setminus \{e_t\}$.
            \State \textbf{for each} $f\in \binom{V}{2}$ and $c\in \Calg$: \; Set $P^{(t)}_{fc}\gets P^{(t-1)}_{fc}$. \Comment{May be overridden below if $f$ neighbors $e_t$}

\color{green!50!black}

            \If{$u$ or $v$ are \emph{bad} (see \cref{def:bad_vertices_etc})}
                \If{$P^{(t-1)}_{ce_t} = 0$ for all $c\in \Calg$, or either $u$ or $v$ are \emph{dangerous} (see \cref{def:bad_vertices_etc})}
                     Mark $e_t$.
                    \label{line:obliv:mark_bad}
                \Else 
                    \State Assign to $e_t$ an arbitrary color $c \in \Calg$, for which $P^{(t-1)}_{e_tc} > 0$.
                    \label{line:obliv:color_bad}
                    \For{\textbf{each} $f \in F_t$ such that $f \cap e_t \neq \emptyset$}
                        \State $P^{(t)}_{fc} \leftarrow 0$.
                        \label{line:obliv:burn}
                    \EndFor
                \EndIf
                \Else
                
            \color{black}
            
            \If{$\sum_{c\in \Calg} P^{(t - 1)}_{e_tc} > 1$}
                Mark $e_t$.
                \label{line:obliv:mark_z_ge_1}
            \Else
                \State Sample $K_t$ from $\Calg \cup \{\perp\}$ with probabilities $\left(P^{(t-1)}_{e_t1},\dots,P^{(t-1)}_{e_t\Delta}, 1 - \sum_{c\in \Calg}P^{(t-1)}_{e_tc}\right)$.
                \label{line:obliv:sample}
                \If{$K_t\in \Calg$} Assign color $K_t$ to $e_t$. \label{line:obliv:color}
                \Else\, Mark $e_t$. \label{line:obliv:mark_bot}
                \EndIf
                \For{\textbf{each} $f \in F_t$ such that $f \cap e_t \neq \emptyset$}
                    \State $P^{(t)}_{fK_t} \leftarrow 0$ \textbf{if} $K_t \in \Calg$
                        \Comment{Prevent $f$ from being colored $K_t$}
                    \label{line:obliv:zero}
                        \State $P^{(t)}_{fc} \leftarrow\frac{P^{(t-1)}_{fc}}{1 - P^{(t-1)}_{e_tc}}$
                        \textbf{for each}
                    $c\in \Calg\setminus\{ K_t\}$ \textbf{where} {\color{blue} $P^{(t-1)}_{fc} \leq A$ ($\star$)}
                    \label{line:obliv:scaleup}
                    \label{line:obliv:star}
                \EndFor
                
            \EndIf
                           \color{green!50!black}
                \If{$e_t$ marked in \cref{line:obliv:mark_z_ge_1} or \cref{line:obliv:mark_bot}}
                    \State Increment $\udeg(u)$ and $\udeg(v)$. \label{line:obliv:increment_badness}
                    \EndIf
            \color{black} 

            \EndIf
            \If{$e_t$ marked}
            \State Color $e_t$ using the greedy algorithm (\cref{lem:greedy}) with the separate palette $\Cgreedy$.
                \label{line:obliv:greedy}
            \EndIf
		\EndFor
	\end{algorithmic}
\end{algorithm}

The algorithm is broadly similar to \cref{alg:edge-coloring}, but contains some further logic depending on whether the arriving edge $e_t$ is incident to any \emph{bad} vertex. See \cref{alg:edge-coloring-oblivious} for the pseudo-code, with all additional logic and modifications (compared to \cref{alg:edge-coloring}) marked in {\color{green!50!black}green}.

Let us therefore first explain what \emph{bad} vertices are, intuitively. Recall that the algorithm \emph{marks} arriving edges (in \cref{line:obliv:mark_bot,line:obliv:mark_z_ge_1}), if $Z_{e_t} > 1$ or $K_{t} = \perp$ respectively, and these edges are forwarded to the emergency palette $\Cgreedy$. In the context of the new \cref{alg:edge-coloring-oblivious}, if a vertex $v$ is incident to such an edge, we increase its \emph{badness} (in \cref{line:obliv:increment_badness}), and we call it \emph{bad} as soon as its \emph{badness} exceeds $2c_K\varepsilon \Delta \approx 1000\eps \Delta$ (see \cref{def:bad_vertices_etc}).

The difference compared to \cref{alg:edge-coloring} is the additional logic treating the case in which the arriving edge is incident to (at least one) \emph{bad} vertex. If so, the edge is treated differently and, in particular, it might be \emph{marked} in \cref{line:obliv:mark_bad}. Importantly, edges \emph{marked} in \cref{line:obliv:mark_bad} do not increase the \emph{badness} of any vertex. However, our analysis will prove that this case is actually negligible (\cref{lemma:second_main_lemma,lemma:bad_vertex_lemma}), and so it turns out that the \emph{badness} is a very good approximation of the number of edges incident at $v$ that are forwarded to the emergency palette~$\Cgreedy$.

Let us now discuss more carefully the logic that handles \emph{bad} vertices. Assume first that neither $u$ nor $v$ (where $e_t := \{u,v\}$) are \emph{dangerous}. This means that neither of these vertices had more than $\alpha \Delta$ \emph{bad} neighbors up to this point (see \cref{def:bad_vertices_etc}). If so, the algorithm tries to color the edge $e_t$ using the palette $\Calg$. If $Z^{(t-1)}_{e_t} > 0$, this is possible, since there exists $c \in \Calg$ such that $P^{(t-1)}_{e_tc} > 0$, which is only possible if $c$ is available for both $u$ and $v$. In this case, such a color $c$ is picked arbitrarily in \cref{line:obliv:color_bad}, and then ``burned'' for all other potential future incident edges $f$ by setting $P_{fc}\gets 0$ in \cref{line:obliv:burn}. Importantly, we do not scale up the $P_{fc'}$ for colors $c'\neq c$, unlike what we do in \cref{line:obliv:scaleup}. If $Z^{(t-1)}_{e_t} = 0$, then this fix does not work and so the edge is \emph{marked} to be forwarded to the emergency palette $\Cgreedy$. Finally, if $u$ or $v$ are \emph{dangerous}, then the plan is abandoned right away, and the edge is immediately forwarded to $\Cgreedy$. This is done to guarantee that we do not ``burn'' too often colors from $\Calg$ in \cref{line:obliv:burn}.
Our analysis will show that: (i) likely there are no \emph{dangerous} vertices (\cref{lemma:second_main_lemma}), and (ii) once a vertex becomes \emph{bad}, the number of edges $e_t$ incident to it such that $Z^{(t-1)}_{e_t} = 0$ is negligible (\cref{lemma:bad_vertex_lemma}).

The following lemma formalizes the above ideas.

\subsection{Analysis} \label{sec:analysis_oblivious}

We begin by showing that \cref{alg:edge-coloring-oblivious} is a valid edge coloring algorithm:
\begin{lemma}
    \cref{alg:edge-coloring-oblivious} provides a valid edge coloring on any fixed input graph.
\end{lemma}
\begin{proof}
    Fix a time step $t$. If the edge $e_t$ is marked, the algorithm guarantees the use of a valid color from the emergency palette $\Cgreedy$ when running greedy.

    It remains to argue the case in which $e_t$ is not marked. We note that a color $c$ can be assigned to $e_t$ only if $P^{(t-1)}_{e_tc} > 0$. Assume by contradiction that the color $c \in \Calg$ assigned to $e_t$ is not valid. Then, there exists an edge $f$, which arrived at a time $t' < t$ (before $e_t$), and which was assigned color $c \in \Calg$. Since $f$ was not marked (otherwise it would have received a color from $\Cgreedy$), either \cref{line:obliv:burn} or \ref{line:obliv:zero} were executed by the algorithm in the iteration corresponding to the arrival of $f$, setting $P^{(t')}_{e_t c}$ to $0$. In particular this implies $P^{(t-1)}_{e_t c} = 0$, which means that $c$ cannot be sampled on the arrival of $e_t$, contradiction.
\end{proof}

In the following we make the central claim of our analysis, which we call the ``Few Bad Vertices Lemma''. It asserts that, with high probability, no vertex is \emph{dangerous}, i.e., no vertex has more than a tiny amount of roughly $o(\Delta)$ \emph{bad} neighbors. This means that, for any \emph{good} vertex $v$, there will be few edges incident to $v$ for which we run the additional logic added to handle \emph{bad} vertices in \cref{alg:edge-coloring-oblivious}. This in turn implies, intuitively, that the number of edges per \emph{good} vertex forwarded to $\Cgreedy$ due to being marked in \cref{line:obliv:mark_bad} is negligible.
\begin{restatable}[Few Bad Vertices Lemma]{lemma}{SecondMainLemma} \label{lemma:second_main_lemma}
    With high probability in $n$, at all time steps $t$ during the execution of \cref{alg:edge-coloring}, no vertex $v \in V$ is \emph{dangerous}.
\end{restatable}

We provide the proof of the lemma in the separate \cref{sec:arguing_second_main_lemma}. We continue by claiming the following (final) so-called ``Bad Vertex Lemma'' which will allow in the end to prove \cref{thm:oblivious_theorem}:
\begin{restatable}[Bad Vertex Lemma]{lemma}{BadVertexLemma} \label{lemma:bad_vertex_lemma}
    With high probability in $n$, for each vertex $v$ that turns bad at time step $t_0$, there are at most $\varepsilon \Delta$ neighbors $u$ of $v$, arriving after time step $t_0$, such that $Z^{(t_{e_t} - 1)}_{e_t} = 0$, where $e_t = \{u,v\}$ and $t_{e_t} > t_0$ is the time step at which $e_t$ arrives.
\end{restatable}
The lemma has this particular name because it will allow to argue that, once a vertex becomes bad, most of its future arriving edges will be colored in \cref{line:obliv:color_bad} (which only works for arriving edges $e_t$ such that $Z^{(t_{e_t}-1)}_{e_t} > 0$), thus allowing to conclude that even \emph{bad} vertices will have few incident edges marked in \cref{line:obliv:mark_bad}. The proof of this lemma is provided in the separate \cref{sec:arguing_bad_vertex_lemma}.

Using the Few Bad Vertices \cref{lemma:second_main_lemma} and the Bad Vertex \cref{lemma:bad_vertex_lemma} we can prove that \cref{alg:edge-coloring-oblivious} successfully colors any input graph using few colors:
\begin{theorem}
    With probability at least $1 - 1/n$, \cref{alg:edge-coloring-oblivious} colors all edges of the input graph using at most
    \begin{equation*}
        O(\varepsilon \Delta + \alpha\Delta) = o(\Delta) \textrm{ colors from $\Cgreedy$.}
    \end{equation*}
\end{theorem}
\begin{proof}
    Let $G' \subseteq G$ be the subgraph of the input graph $G$ which contains all marked edges. It suffices to prove that the maximum degree of $G'$ is, with high probability in $n$, bounded by $O(\varepsilon \Delta + \alpha\Delta)$. We argue that, for each fixed vertex $v \in V$, its degree in $G'$ is, with high probability in $n$, bounded by this quantity. (Then we obtain the desired statement by a union bound over all $n$ vertices.)

    First, we claim that, for any time step $t$ at which $v$ is \emph{good}, $v$ has at most $2c_K\varepsilon \Delta + \alpha \Delta$ incident marked edges. Indeed, the edges incident to $v$ up to time step $t$ are marked either at \cref{line:obliv:mark_bot} or \ref{line:obliv:mark_z_ge_1} in \cref{alg:edge-coloring-oblivious}, which increases the \emph{badness} of $v$, or at \cref{line:obliv:mark_bad} in \cref{alg:edge-coloring-oblivious}. The former case can occur at most $2c_K\varepsilon \Delta$ times, since $v$ is \emph{good} and thus by assumption has bounded \emph{badness} (\cref{def:bad_vertices_etc}). The latter case can occur only for edges connecting $v$ to already \emph{bad} neighbors, and by \cref{lemma:second_main_lemma} we know there are, with high probability in $n$, at most $\alpha \Delta$ such edges incident to $v$.

    Second, we claim that, once a vertex turns \emph{bad}, there can be (with high probability in $n$) at most $O(\varepsilon \Delta)$ further marked edges incident to it. To see this, first note that edges incident to a \emph{bad} vertex can only be marked in \cref{line:obliv:mark_bad} of \cref{alg:edge-coloring-oblivious}. Furthermore, we know (\cref{lemma:second_main_lemma}) that, with high probability, no vertex in the input graph is ever \emph{dangerous}, such that these marked edges $e_t$ must have $Z^{(t-1)}_{e_t} = 0$ for \cref{line:obliv:mark_bad} to be reachable. \cref{lemma:bad_vertex_lemma} implies that there are at most $O(\varepsilon \Delta)$ such edges incident at $v$, which proves the claim.

    To sum up, we proved (with high probability) that, while a vertex is \emph{good}, it is incident to at most $2c_K\varepsilon\Delta + \alpha\Delta$ marked edges and, once it becomes \emph{bad}, this number can increase by at most $O(\varepsilon \Delta)$. Thus, the statement of the theorem follows.
\end{proof}

\subsubsection{Adaptations of Previous Results for \cref{alg:edge-coloring-oblivious}}

Before presenting the proofs of \cref{lemma:second_main_lemma} and \cref{lemma:bad_vertex_lemma}, our first goal is to recover or extend several of the results already proven in \cref{sec:adaptive}. We recall that the $\Delta+o(\Delta)$-coloring guarantee proven for \cref{alg:edge-coloring} holds under the assumption that $\Delta = \omega(\log n)$. However, several useful properties proven in the analysis of \cref{alg:edge-coloring} in \cref{sec:analysis} hold for \cref{alg:edge-coloring-oblivious} even under the weaker hypothesis that $\Delta = \omega(\sqrt{\log n})$, in the case of an oblivious (instead of adaptive) adversary. We discuss in this section those properties which will later be relevant in the analysis.

First, we note that \cref{lemma:invariant_ub_on_P} continues to hold trivially. Next, recall the Few Bad Colors \cref{lemma:main_lemma} from the analysis of \cref{alg:edge-coloring} (recall that $\varepsilon$ was defined differently in that context):
\MainLemma*
The crucial part of the proof (in \cref{sec:arguing_main_lemma}) of the above lemma was arguing that $n^{\Theta(\Delta)}$ many martingales were each concentrated enough with probability $1-\exp(-\poly(\eps)\Delta^2)$. In the adaptive adversary setting, we used the fact that $\Delta = \Omega(\log n)$ to argue that this is enough for taking a union bound over all these $n^{\Theta(\Delta)}$ many martingales, representing the potential future neighborhoods the adaptive adversary could play.

In the oblivious setting, we can use the same analysis for the martingales and again obtain that each is concentrated with probability at least $1-\exp(-\poly(\eps)\Delta^2)$. The advantage is that, unlike the adaptive adversary, there is only one fixed future we need to argue about for the oblivious adversary. Thus, in the oblivious setting, it suffices to argue that only $O(n^2)$ many such martingales are concentrated (one for each edge), much fewer than the $n^{\Theta(\Delta)}$ many in \cref{sec:arguing_main_lemma}.
This means that $\Delta = \Omega(\sqrt{\log n})$ is already enough for the necessary concentration, in the oblivious setting.

We thus claim that the same Few Bad Colors Lemma holds in the new oblivious setting, using \cref{alg:edge-coloring-oblivious} and the different parametrization of $\varepsilon$. In fact, here we prove slightly stronger probability bounds than ``with  high probability in $n$'' so that we are allowed to union-bound over $2^{O(\Delta)}\cdot \poly(n)$ many events later.
\begin{restatable}[Few Bad Colors for Oblivious Setting]{lemma}{MainLemmaOblivious} \label{lemma:main_lemma_oblivious}
    Consider running \cref{alg:edge-coloring-oblivious} in the aforementioned oblivious adversary setting. With probability at least $1-2^{-100\Delta}n^{-100}$, for all time steps $t$ such that $e \in F_t$, we have:
    \begin{equation*}
        \left| B^{(t)}_e \right| \leq 2\varepsilon^5 \Delta,
    \end{equation*}
    where $B^{(t)}_e = \left\{c \in \Calg: P^{(t)}_{ec} > A \right\}$ is the set of \emph{bad} colors w.r.t.\ $e$ at time $t$.
\end{restatable}
\begin{proof}[Proof Sketch]
    We need to address the changes caused by the different definition of $\varepsilon$, the fact that the adversary is now oblivious, the change of \cref{alg:edge-coloring-oblivious} compared to \cref{alg:edge-coloring}, and the fact that $\Delta = \omega(\sqrt{\log n})$ instead of the previous stronger assumption. By inspecting the proof of the Few Bad Colors Lemma in \cref{sec:arguing_main_lemma}, keeping in mind the stated differences, we note that \cref{lemma:ub_on_P_in_terms_of_S,,lemma:relating_S_to_Q,corollary:characterization_of_bad_colors} remain valid.
    
    Since the neighborhood $U_e$ of an edge $e$ is now fixed (because we consider oblivious adversaries instead of adaptive ones; the following is just for the analysis and the algorithm itself does not need knowledge of the fixed $U_e$), we now claim that to obtain \cref{lemma:main_lemma_oblivious} it suffices to prove the following analogon of \cref{lemma:bounds_on_Q}:
    \begin{lemma} \label{lemma:bounds_on_Q_oblivious}
        Let $U_e = U_u \cup U_v$ be the \emph{neighborhood} for edge $e$ in a fixed instance of the online edge coloring problem. Let $w \in \{u,v\}$, and consider a subset of $2 \varepsilon^5 \Delta$ colors $C \subseteq \Calg$. Then, we have:
        \begin{equation*}
            \Pr\left[ Q^{(t_e-1)}_{U_wC} \geq \varepsilon^5\Delta - \frac{\varepsilon^6\Delta}{2} \right] \leq \frac{1}{2^{110\Delta}n^{110}}.
        \end{equation*}
    \end{lemma}
    To see why proving this sub-lemma suffices, we mimic the proof of \cref{lemma:main_lemma} from \cref{sec:arguing_main_lemma_part2}. Let $\binom{\Calg}{\varepsilon^5 \Delta}$ be the set of all subsets of $\varepsilon^5 \Delta$ colors from $\Calg$. Then, one can see that:
    \begin{align*}
        \Pr\left[ \left|B^{(t)}_{e}\right| > 2 \varepsilon^5 \Delta \right] \leq \sum_{C \in \binom{\Calg}{\varepsilon^5 \Delta}}\sum_{w \in \{u,v\}} \Pr\left[ \sum_{c \in C}Q^{(t_e-1)}_{U_wc} > \varepsilon^5\Delta - \frac{\varepsilon^6\Delta}{2}\right].
    \end{align*}
    Upper bounding the relevant cardinalities:
    \begin{equation*}
        \left| \binom{\Calg}{\varepsilon^5 \Delta} \right| \cdot \left|\{u,v\}\right| \leq 2^{\Delta} \cdot 2
    \end{equation*}
    and using \cref{lemma:bounds_on_Q_oblivious}, which holds for any $C$ and $w$, we obtain that, as desired:
    \begin{align*}
        \Pr\left[ \left|B^{(t)}_{e}\right| > 2 \varepsilon^5 \Delta \right] & \leq \frac{1}{2^{100\Delta}n^{100}}. \qedhere
    \end{align*}
\end{proof}
\begin{proof}[Proof of \cref{lemma:bounds_on_Q_oblivious}]
    Similarly to the proof of \cref{lemma:bounds_on_Q}, using the fact that $Q_{U_wC}$ is (upper bounded by) a supermartingale, it turns out that it suffices to argue:
    \begin{equation*}
        \frac{1}{(c_A \cdot 48)^2} \cdot \varepsilon^{16}\Delta^2 \geq 110\ln n + 110\Delta.
    \end{equation*}
    Recalling $\varepsilon = c_\varepsilon \cdot (\sqrt{\ln n} / \Delta)^{1/16}$, this is equivalent to:
    \begin{equation*}
        \frac{1}{(c_A \cdot 48)^2} \cdot c^{16}_{\varepsilon} \cdot \Delta \sqrt{\ln n} \geq 110\ln n +110\Delta,
    \end{equation*}
    which holds by the choices of the constants $c_A$, $c_\varepsilon$, and the fact that $\Delta \geq \sqrt{\ln n}$.
\end{proof}

In the following we provide extensions of \cref{lemma:Zlbhelperlemma1} and \cref{lemma:Zlbhelperlemma2} from the previous analysis in \cref{sec:analysis}, which will be useful for the later analysis, in particular \cref{lemma:matching_lemma}. We include these preliminary lemmas in this section to avoid cluttering future arguments. Before stating the lemmas, we need some further notations and definitions, and in particular to adapt \cref{def:lb_on_Z} from \cref{sec:analysis}.

In the following, fix an instance of the online edge coloring problem, and an edge $e = \{u,v\}$ with arrival time $t_e$. Consider the edges $\delta(e)$ that intersect with $e$, and which arrive before $e$.
\begin{definition} \label{def:classifying_edges}
    We partition the edges of $\delta(e)$ into three subsets $\delta(e)_{\text{good}}$, $\delta(e)_{\text{bad}}$ and $\delta(e)_{\text{rest}}$, as follows:
    \begin{itemize}
        \item It holds that $f \in \delta(e)_{\text{good}}$ if, at the time of arrival, none of the two vertices to which $f$ is incident are \emph{bad}.
        \item It holds that $f \in \delta(e)_{\text{bad}}$ if $f$ is incident to a bad vertex on arrival, but this bad vertex is not $u$ or $v$.
        \item It holds that $f \in \delta(e)_{\text{rest}}$ if $f$ is incident to a bad vertex on arrival, and this vertex is $u$ or $v$.
    \end{itemize}
\end{definition}
Further, we extend slightly the definition of $Z_e$:
\begin{definition} \label{def:Z_extension}
    Fix an arbitrary subset $C \subseteq \Calg$. Define, for any time step $t$, $Z^{(t)}_{eC} = \sum_{c \in C} P^{(t)}_{eC}$.
\end{definition}
Finally, we extend \cref{def:lb_on_Z} from \cref{sec:analysis} to define the variables $Y_{eC}$ associated to $Z_{eC}$:
\begin{definition} \label{def:lb_on_Z_extended}
For each $f\in F_{0}$, color $c\in \Calg$, and time step $t$, let
\begin{equation*}
\overline{P^{(t)}_{fc}}
=
\begin{cases}
P^{(t-1)}_{fc} / (1-P^{(t-1)}_{e_tc}) & \text{if in \Cref{line:obliv:scaleup} no assignment to $P^{(t)}_{fc}$ is made because of ${\color{blue}P^{(t-1)}_{fc} > A\quad (\star)}$}\\
P^{(t)}_{fc} & \text{otherwise.}
\end{cases}
\end{equation*}
    Define also, for any fixed subset $C \subseteq \Calg$:
    \begin{equation*}
        \overline{Z^{(t)}_{eC}} = \sum_{c \in \C} \overline{P^{(t)}_{ec}} \text{ \ \ and \ \ } Y^{(t)}_{eC} 
        = \sum_{t'=1}^t \mathds{1}[e_{t'} \in \delta(e)_{\text{good}}] \cdot \left( \overline{Z^{(t')}_{eC}} - Z^{(t'-1)}_{eC} \right).
    \end{equation*}
\end{definition}
As before in the case of \cref{def:lb_on_Z}, the variables $\overline{P^{(t)}_{ec}}$ ar per-time-step alternative versions of $P^{(t)}_{ec}$ in which the condition $\starcond$  which potentially limits the scaling is ignored. Since this scaling step is only done for edges which are not incident to any \emph{bad} vertex, we also modify the definition of the previous $Y_e$ accordingly, to only consider the sum over such arrivals. Another change is given by the slight generalization to sums defined over arbitrary subsets of colors.

We are now ready to provide the claimed extensions of \cref{lemma:Zlbhelperlemma1} and \cref{lemma:Zlbhelperlemma2}:
\begin{lemma} \label{lemma:Zlbhelperlemma1_extended}
     The sequence of random variables $Y^{(0)}_{eC}, Y^{(1)}_{eC},\dots$ is a martingale with respect to the sequence of random variables $K_1,K_2,\dots$ Furthermore, the step size is bounded by $6A$.
\end{lemma}
\begin{proof}
    Since $Y^{(t)}_{eC}$ can have non-zero changes only for arriving edges $e_{t'} \in \delta(e)_{\text{good}}$, for which \cref{alg:edge-coloring-oblivious} is identical to \cref{alg:edge-coloring}, we can reuse the previous analysis and thus the proof of \cref{lemma:Zlbhelperlemma1} in \cref{sec:proofs_Zlbhelperlemmas}. (Also, the fact that the sum is taken over a subset of the colors from $\Calg$ does not change anything.)
\end{proof}
\begin{lemma} \label{lemma:Zlbhelperlemma2_extended}
    We have, with probability at least $1-2^{-100\Delta}n^{-100}$, for any time step $t$:
    \begin{equation*}
        \frac{|C|}{\Delta}(1- \varepsilon) + Y^{(t)}_{eC} - 20 \cdot \varepsilon^5 \Delta^2 A^2 - (4\alpha\Delta + 2|\delta(e)_{\text{rest}}|)\cdot A \leq Z^{(t)}_{eC} \leq \frac{|C|}{\Delta}(1-\varepsilon) + Y^{(t)}_{eC}.
    \end{equation*}
\end{lemma}
\begin{proof}
    Note that, using $Z^{(0)}_{eC} = \frac{|C|}{\Delta}(1-\varepsilon)$, we have:
    \begin{equation} \label{eq:aux_identity_on_ZeC}
        Z^{(t)}_{eC} = \frac{|C|}{\Delta}(1-\varepsilon) + \sum_{t'=0}^{t-1} \mathds{1}[e_{t'} \in \delta(e)_{\text{good}}]\left( Z^{(t'+1)}_{eC} - Z^{(t')}_{eC} \right) +  \sum_{t'=0}^{t-1} \mathds{1}[e_{t'} \notin \delta(e)_{\text{good}}]\left( Z^{(t'+1)}_{eC} - Z^{(t')}_{eC} \right).
    \end{equation}
    To obtain the upper bound, notice that we may ignore the last sum of the right hand side. This is because edges $e_{t'} \notin \delta(e)_{\text{good}}$ can only lead to a decrease in $Z_{eC}$. Next, we trivially have:
    \begin{align*}
        Z^{(t)}_{eC} &\leq \frac{|C|}{\Delta}(1-\varepsilon) + \sum_{t'=0}^{t-1} \mathds{1}[e_{t'} \in \delta(e)_{\text{good}}]\left( Z^{(t'+1)}_{eC} - Z^{(t')}_{eC} \right) \\
        &\leq  \frac{|C|}{\Delta}(1-\varepsilon) + Y^{(t)}_{eC},
    \end{align*}
    which gives the claimed upper bound.
    
    For the lower bound, we begin by focusing on the last sum of the right hand side in \cref{eq:aux_identity_on_ZeC}. Note that:
    \begin{align*}
        &\sum_{t'=0}^{t-1} \mathds{1}[e_{t'} \notin \delta(e)_{\text{good}}]\left( Z^{(t'+1)}_{eC} - Z^{(t')}_{eC} \right) \\
        =&\sum_{t'=0}^{t-1} \mathds{1}[e_{t'} \in \delta(e)_{\text{bad}}]\left( Z^{(t'+1)}_{eC} - Z^{(t')}_{eC} \right) + \sum_{t'=0}^{t-1} \mathds{1}[e_{t'} \in \delta(e)_{\text{rest}}]\left( Z^{(t'+1)}_{eC} - Z^{(t')}_{eC} \right)
    \end{align*}
    For any edge $e_{t'}$ incident to a \emph{bad} vertex (be that $u$, $v$, or another vertex), \cref{alg:edge-coloring-oblivious} might zero out a variable $P_{eC}$ in \cref{line:obliv:burn}, leading to a decrease in $Z_{eC}$ of at most $2A$ (by \cref{lemma:invariant_ub_on_P}). Furthermore, once both $u$ and $v$ become \emph{dangerous} (see \cref{def:bad_vertices_etc}), the value of $Z_{eC}$ is frozen. Hence, there are at most $2\alpha\Delta$ edges $e_{t'} \in \delta(e)_{\text{bad}}$ which we need to sum over in the above first sum. By this argumentation, we obtain that:
    \begin{align*}
        \sum_{t'=0}^{t-1} \mathds{1}[e_{t'} \notin \delta(e)_{\text{good}}]\left( Z^{(t'+1)}_{eC} - Z^{(t')}_{eC} \right) \geq - (4\alpha\Delta + 2|\delta(e)_{\text{rest}}|)\cdot A.
    \end{align*}
    Turning back to \eqref{eq:aux_identity_on_ZeC}, we obtain:
    \begin{align*}
        Z^{(t)}_{eC} &\geq \frac{|C|}{\Delta}(1- \varepsilon) - (4\alpha\Delta + 2|\delta(e)_{\text{rest}}|)\cdot A + \sum_{t'=0}^{t-1} \mathds{1}[e_{t'} \in \delta(e)_{\text{good}}]\left( Z^{(t'+1)}_{eC} - Z^{(t')}_{eC} \right) \\
        &=\frac{|C|}{\Delta}(1- \varepsilon) + Y^{(t)}_{eC} - (4\alpha\Delta + 2|\delta(e)_{\text{rest}}|)\cdot A - \sum_{t'=0}^{t-1} \mathds{1}[e_{t'} \in \delta(e)_{\text{good}}]\left( \overline{Z^{(t'+1)}_{eC}} - Z^{(t'+1)}_{eC} \right),
    \end{align*}
    such that it remains to upper bound the sum:
    \begin{equation*}
        S_{\overline{Z},Z} = \sum_{t'=0}^{t-1} \mathds{1}[e_{t'} \in \delta(e)_{\text{good}}]\left( \overline{Z^{(t'+1)}_{eC}} - Z^{(t'+1)}_{eC} \right).
    \end{equation*}
    This is the ``drift'' downwards of the supermartingale $Z_{eC}$ coming from the bad colors.
    Since the sum is taken over edges with no bad endpoints,  \cref{alg:edge-coloring-oblivious}'s logic is identical to \cref{alg:edge-coloring}, and upper bounding $S_{\overline{Z}, Z}$ can be done exactly as in the proof of \cref{lemma:Zlbhelperlemma2}, giving an $20 \cdot \varepsilon^5 \Delta^2 A^2$ upper bound, now with probability at least $1-2^{-100\Delta}n^{-100}$ (since we are using \cref{lemma:main_lemma_oblivious} instead of \cref{lemma:main_lemma}).
\end{proof}

\subsection{Few Bad Vertices: Proof of \cref{lemma:second_main_lemma}} \label{sec:arguing_second_main_lemma}

For convenience, we recall the main result that is to be proven in this section:
\SecondMainLemma*

Before presenting the formal proof, we provide a brief outline. The plan is to show that any subset $U$ of $|U| = \alpha \Delta$ vertices has bounded total \emph{badness}, meaning that the quantity $\sum_{u \in U} \udeg(u)$ is not too large. If the proven bound on the total \emph{badness} holds with probability $1 - \exp(- \poly(\varepsilon) \Delta^2)$ for any such fixed subset $U$ of $\alpha \Delta$ vertices, the statement of the lemma follows by taking a union bound over all \emph{at most} $2^\Delta \ll \exp(\poly(\varepsilon) \Delta^2)$ possible subsets of $\alpha \Delta$ neighbors of the vertex $v$. 

To obtain the desired concentration for the bound on the total \emph{badness} $\sum_{u \in U} \udeg(u)$, we would ideally like to first claim an upper bound $\udeg(u) \leq O(\varepsilon \Delta)$ for each individual $u \in U$ (i.e., according to \cref{def:bad_vertices_etc}, $u$ is not \emph{bad}), with probability say $1 -\exp(-\poly(\varepsilon)\Delta)$, and then obtain by independence with Chernoff bounds that $\sum_{u \in U} \udeg(u) \leq O(\varepsilon \Delta^2)$ with probability $1 - \exp(-\poly(\varepsilon)\Delta^2)$. However, as usual in our analysis, the existence of correlations does not allow for this simple argument, and we need to circumvent it by an alternative martingale argument. This argument itself turns out to require further care, as we will explain later, since defining a ``naive'' martingale over all edges incident to vertices in $U$ might lead to an uncontrolled step size.

We now resume with the formal proof of \cref{lemma:second_main_lemma}. The proof relies on the following lemma:
\begin{restatable}{lemma}{BoundingZPerMatching} \label{lemma:bounding_Z_per_matching}
    Consider a fixed instance of the online edge coloring problem, and let $M$ be an arbitrary matching in the input graph, of size $|M| \leq \alpha \Delta$. Then, with probability at least $1 - 2^{-110\Delta}n^{-110}$, there are at most $2 \varepsilon \alpha \Delta$ edges $e$ in $M$ such that: no endpoint of $e$ is \emph{bad} at arrival time $t_e$ of $e$, and $Z^{(t_e-1)}_{e} \notin [1-c_K \cdot \varepsilon, 1]$, where $c_K = 35c^2_A \leq 1000$ is a constant.
\end{restatable}
Before discussing the proof of this latter lemma, we show how the Few Bad Vertices Lemma follows:
\begin{proof}[Proof of \cref{lemma:second_main_lemma}]
    Our goal is to prove that, with high probability in $n$, no vertex $v$ has more than $\alpha \Delta$ neighbors which are \emph{bad} by the time they are connected to $v$ (see \cref{def:bad_vertices_etc}). In fact, we prove a slightly stronger statement: fix an arbitrary subset $U$ of $|U| = \alpha \Delta$ vertices. We claim that, for any time step $t$:
    \begin{equation} \label{eq:sufficient_cond_second_main_lemma}
        \Pr \left[ \sum_{u \in U} \udeg^{(t)}(u) \geq 2c_K \cdot  \varepsilon \alpha \Delta^2 \right] \leq \frac{1}{2^{100\Delta}n^{100}}.
    \end{equation}
    To see why this suffices, let $\binom{N(v)}{\alpha\Delta}$ be the set of all subsets of $\alpha\Delta$ neighbors of $v$, and note that there are at most $\binom{\Delta}{\alpha\Delta}\le 2^{\Delta}$ such subsets. We obtain
    \begin{align*}
        \Pr\left[\text{$v$ dangerous}\right] &\leq \Pr\left[ \text{$\exists U \in \binom{N(v)}{\alpha\Delta}$ s.t. $\exists t \in [1,n^2]$ s.t. $\forall u \in U : \udeg^{(t)}(u) \geq 2c_K \cdot \varepsilon \Delta$} \right] \\
        &\leq \Pr\left[ \text{$\exists U \in \binom{N(v)}{\alpha\Delta}$ s.t. $\exists t \in [1,n^2]$ s.t. $\sum_{u \in U} \udeg^{(t)}(u) \geq 2c_K \cdot \varepsilon \alpha \Delta^2$} \right] \\
        &\leq \sum_{U \in \binom{N(v)}{\alpha\Delta}} \sum_{t \in [1,n^2]} \Pr \left[ \sum_{u \in U} \udeg^{(t)}(u) \geq 2 c_K \cdot \varepsilon \alpha \Delta^2 \right] \\
        &\stackrel{\eqref{eq:sufficient_cond_second_main_lemma}}{\leq} 2^{\Delta} \cdot n^2 \cdot \frac{1}{2^{100\Delta}n^{100}}
        \\
        &\leq \frac{1}{n^{90}}.
    \end{align*}
    To prove \eqref{eq:sufficient_cond_second_main_lemma}, we consider the subgraph $G_U \subseteq G$ of the input graph which contains those edges $E_U \subseteq E$ incident to vertices in $U$. Since $G_U$ has maximum degree $\Delta$, its edge set $E_U$ can be partitioned into at most $\Delta + 1$ matchings $M_1,\dots,M_{\Delta + 1}$. Further, each such matching $M_i$ can have size at most $\alpha \Delta$, because all edges in $G_U$ are incident to one of the $\alpha \Delta$ vertices in $U$.

    We analyze the evolution of the \emph{badness} of vertices in $U$ as the $|E_U| \leq \alpha \Delta^2$ edges from $E_U$ arrive. There are different cases to consider, depending on the status of the arriving edge $e_t$ at the time $t_{e_t}$ of its arrival:
    \begin{enumerate}
        \item If $e_t$ is incident to a \emph{bad} vertex on arrival, \cref{alg:edge-coloring-oblivious} does not increase the \emph{badness} of any vertex, such that this case is trivial.
        
        \item Otherwise, if $e_t$ connects two \emph{good} vertices, it might be the case that $Z^{(t_{e_t} - 1)}_{e_t} \notin [1-c_K \cdot \varepsilon, 1]$. Each such arrival can increase the sum $\sum_{u \in U} \udeg(u)$ by at most $2$, due to \cref{line:obliv:increment_badness}. By \cref{lemma:bounding_Z_per_matching}, with probability at least $1 - 2^{-110\Delta}n^{-110}$, this happens for at most $2 \varepsilon \alpha \Delta$ edges per matching $M_i$, implying the increase in \emph{badness} is of at most $5 \cdot \varepsilon \alpha \Delta (\Delta+1) \leq 6\varepsilon \alpha \Delta^2$ in total.
        
        \item It remains to analyze the case of edges $e_t$ connecting two \emph{good} vertices, and $Z^{(t_{e_t} - 1)}_{e_t} \in [1- c_K \cdot \varepsilon, 1]$. Let $F$ be the set of these edges. For each $f \in F$, define the random indicator $X_f$ which is $1$ iff $f$ leads to an increase of \emph{badness} (by $2$) in \cref{line:obliv:increment_badness}. Since $Z_f \in [1- c_K \cdot \varepsilon,1]$ at the time $t_f$ at which $f$ arrives, and $f$ is not incident to any \emph{bad} vertex, we have $\Pr[X_f = 1] \leq c_K \cdot \varepsilon$. This is because the probability that $f$ is marked is given by the probability that no valid color is sampled (i.e., $K_{t_f} = \perp$) for $f$ in \cref{line:obliv:sample}, and this is given by $Z^{(t_f-1)}_{f} \leq c_K \cdot \varepsilon$.

        Let $f_1,\dots,f_{|F|}$ be an ordering of the edges in $F$ by their arrival time. Our arguments above have shown that, for any $f_i \in F$, we have $\Pr[X_{f_i} \mid X_{f_1},\dots,X_{f_{i-1}}] \leq c_K \cdot \varepsilon$, since the decision w.r.t.\ coloring edge $f$ is independent of previous such decisions. Due to $|F| \leq |E_U| = \alpha \Delta^2$, we also have $\E[\sum_{f\in F} X_f] \leq c_K \cdot \varepsilon\alpha \Delta^2$. Thus, by standard coupling arguments, using Chernoff bounds we obtain that: 
        \begin{align*}
            \Pr \left[ \sum_{f \in F} X_f \geq (c_K + 1) \cdot  \alpha \varepsilon\Delta^2\right] &= \Pr \left[ \sum_{f \in F} X_f \geq c_K \cdot \alpha \varepsilon\Delta^2 \cdot \left( 1 + \frac{1}{c_K} \right) \right] \\
            &\leq \exp \left( - \frac{\alpha \varepsilon\Delta^2}{3 \cdot  c_K} \right) \\
            &\leq \frac{1}{2^{110\Delta}n^{110}}.
        \end{align*}
        Hence, with probability at least $1 - 2^{-110\Delta}n^{-110}$, the increase of \emph{badness} caused by edges considered in this case is at most $2 \cdot 5 \varepsilon \alpha \Delta^2 = 10 \varepsilon \alpha \Delta^2$.
    \end{enumerate}
    By the case distinction above, union bounding over the failures of the Chernoff bound and $\Delta+1$ applications of \cref{lemma:bounding_Z_per_matching}, we obtain that with probability at least $1 - (\Delta+2) 2^{-110\Delta}n^{-110} \geq 1 - 2^{-100\Delta}n^{-100}$, the total increase of \emph{badness} is upper bounded by $6 \varepsilon \alpha \Delta^2 + (c_K+1) \cdot  \varepsilon \alpha \Delta^2 \leq 2c_K \cdot  \varepsilon \alpha \Delta^2$. This implies \eqref{eq:sufficient_cond_second_main_lemma}.
\end{proof}

\paragraph{Arguing about Matchings.}
It remains to argue \cref{lemma:bounding_Z_per_matching}. This is achieved by going through another lemma (see below), which we write in a more general version than strictly needed for our purposes here. However, the general version will be useful in the later proof of the Bad Vertex Lemma (\cref{lemma:bad_vertex_lemma}) in \cref{sec:arguing_bad_vertex_lemma}. 
We also recall \cref{def:classifying_edges}, \cref{def:Z_extension}, \cref{def:lb_on_Z_extended}, and \cref{lemma:Zlbhelperlemma1_extended}, \cref{lemma:Zlbhelperlemma2_extended} from \cref{sec:algorithm_oblivious} since they are now relevant.

As we have seen previously, $Z^{(t)}_{e} = \sum_{c} P^{(t)}_{ec}$ is a supermartingale. This fact can be used to upper bound $Z^{(t)}_e$. Moreover, the drift downwards in $Z^{(t)}_e$ is (with high probability) small (since it does not drop far below the martingale $1-\eps + Y^{(t)}_{e}$, by \cref{lemma:Zlbhelperlemma2_extended}), which can be used to also lower-bound $Z^{(t)}_{e}$.

In \cref{sec:analysis}
(where crucially $\Delta = \Omega(\log n)$)
we used these facts to show that with high probability $Z_{e}\in [1-O(\eps), 1]$.
\emph{Unfortunately, this is not true when $\Delta \ll \log n$, even against a weaker oblivious adversary}. Indeed, if arguing the concentration of $Z_e\in [1-O(\eps),1]$ we can only prove that it holds with probability $1-\exp(-\poly(\eps)\Delta)$, meaning that roughly a $1 / 2^{\sqrt{\log n}}$-fraction of our edges could be outside the range when $\Delta \approx \sqrt{\log n}$ (even $Z_e = 0$ would happen).

Instead, we will resort to arguing that $\sum_{e\in M} Z_e$ is well-concentrated for some set of edges $M$ of size $\approx \poly(\eps) \Delta$. If all the edges in $M$ behaved independently, then we would immediately have concentration $\sum_{e\in M} Z_e \in [|M|(1-\eps), |M|]$ with probability $1-\exp(\poly(\eps)\Delta^2) = 1-\frac{1}{n^{100}}$. However, the $Z_e$ for $e\in M$ might, of course, not be independent from each other, so we follow our usual strategy of analyzing this sum with martingales and Azuma's inequality.

Still, the set of edges $M$ cannot be arbitrary: if they all are incident to the same vertex, they could potentially be very correlated (in more technical terms: the step size in our (super)martingale $\sum_{e\in M} Z_e$ would blow up).
However, if the set of edges $M$ forms a \emph{matching}, then they behave well: any one arriving edge $f$ can touch at most two edges in $M$, allowing us to argue that $\sum_{e\in M} Z_e$ essentially has the same step size as just a single $Z_e$. This motivates the following lemma (where we also restrict our attention to a subset of colors $C$):

\begin{lemma}[Matching Lemma] \label{lemma:matching_lemma}
    Let $M$ be a matching and $C \subseteq \Calg$ be a subset of colors, with the property that $|M| \cdot |C|^2 \geq \varepsilon^{10} \Delta^3$. Let, for any time step $t$, $K^{(t)}_{MC} = \sum_{e \in M} Z^{(t)}_{eC} = \sum_{e \in M} \sum_{c \in C}  P^{(t)}_{ec}$. Then:
    \begin{align*}
        &\Pr\left[K^{(t)}_{MC} > \frac{|M|\cdot |C|}{\Delta}(1-\tfrac{\eps}{2})\right] \le \frac{1}{2^{150\Delta}n^{150}}, \quad \text{and} \\
        &\Pr\left[K^{(t)}_{MC} < \frac{|M|\cdot |C|}{\Delta}-c_K \cdot \eps|M| \text{ and there are no bad endpoints of $M$}\right] \leq \frac{1}{2^{150\Delta}n^{150}},
    \end{align*}
    where $c_K = 35 c_A^2 \leq 1000$ is a fixed constant (see also \cref{sec:notation_assumptions_oblivious}).
\end{lemma}
\begin{proof}
    Let $L^{(t)}_{MC} = \sum_{e \in M} Y^{(t)}_{eC}$. By adding the corresponding upper bound from \cref{lemma:Zlbhelperlemma2_extended} over all edges in $M$, we obtain:
    \begin{equation*}
        K^{(t)}_{MC} \leq \frac{|M| \cdot |C|}{\Delta} (1 - \varepsilon) + L^{(t)}_{MC},
    \end{equation*}
    such that it suffices to prove:
    \begin{equation} \label{eq:aux_toprove_later1}
        \Pr\left[ L^{(t)}_{MC} > \frac{|M| \cdot |C| \cdot \varepsilon}{2\Delta} \right] \leq \frac{1}{2^{150\Delta}n^{150}}.
    \end{equation}
    We come back to proving this inequality later, and focus on the lower bound next. By adding the corresponding lower bound from \cref{lemma:Zlbhelperlemma2_extended} over all edges in $M$, we obtain:
    \begin{equation*}
        K^{(t)}_{MC} \geq \frac{|M| \cdot |C|}{\Delta} (1 - \varepsilon) + L^{(t)}_{MC} - 20|M| \cdot \varepsilon^5 \Delta^2 A^2 - 4|M| \cdot  \alpha\Delta A- 2|M| \cdot  A \sum_{e \in M} |\delta(e)_{\text{rest}}|.
    \end{equation*}
    The lower bound is only claimed in conjunction with the event that $M$ has no bad endpoints. This event in particular implies $|\delta(e)_{\text{rest}}| = 0$ for any $e \in M$ (see \cref{def:classifying_edges}). This allows ignoring the last term in the above right hand side. Furthermore, we have $20|M| \cdot \varepsilon^5 \Delta^2 A^2 \leq 20 c^2_A \cdot \varepsilon |M|$ and $4|M| \cdot \alpha \Delta A \leq \frac{4c_A}{100} \cdot \varepsilon^3 |M| \leq c_A \cdot \varepsilon |M|$. Hence, the above lower bound implies:
    \begin{equation*}
        K^{(t)}_{MC} \geq \frac{|M| \cdot |C|}{\Delta} (1 - \varepsilon)  - 25 c_A^2 \cdot \varepsilon|M| +  L^{(t)}_{MC} \geq \frac{|M| \cdot |C|}{\Delta}- 30c_A^2 \cdot \varepsilon |M| +  L^{(t)}_{MC}.
    \end{equation*}
    It suffices to prove:
    \begin{equation} \label{eq:aux_toprove_later2}
        \Pr\left[ L^{(t)}_{MC} < - \frac{|M| \cdot |C| \cdot \varepsilon}{2\Delta} \right] \leq \frac{1}{2^{150\Delta}n^{150}},
    \end{equation}
    since, if this holds, we obtain immediately $ K^{(t)}_{MC} \geq \frac{|M| \cdot |C|}{\Delta} - 35c_A^2 \cdot \varepsilon |M|$, which gives the claim.

    To prove \eqref{eq:aux_toprove_later1} and \eqref{eq:aux_toprove_later2}, we first note that $L_{MC}$ is a martingale with initial value $0$, because it is a sum of such martingales, by \cref{lemma:Zlbhelperlemma1_extended}. Further, since any arriving edge is incident to at most two edges from the matching $M$, affecting at most two terms of the form $Y_{eC}$ in the sum, the step size of $L_{MC}$ can be at most twice the step size of such a term. By \cref{lemma:Zlbhelperlemma1_extended}, we deduce that the step size of $L_{MC}$ is at most $12A$. To set up Azuma's inequality, take $\lambda = \frac{|M| \cdot |C| \cdot \varepsilon}{2\Delta}$, and note that the martingale $L_{MC}$ takes at most $2\Delta|M|$ steps. Apply \cref{lemma:azuma} to obtain:
    \begin{equation*}
        \Pr \left[ \left| L^{(t)}_{MC}  \right| \geq \lambda \right] \leq 2\exp \left( - \frac{\lambda^2}{4\Delta|M| \cdot (12A)^2}. \right)
    \end{equation*}
    It suffices to argue that:
    \begin{equation*}
        \frac{\lambda^2}{4\Delta|M| \cdot (12A)^2} \geq 150 \ln n +150\Delta+\ln 2.
    \end{equation*}
    Plugging in $\lambda =  \frac{|M| \cdot |C| \cdot \varepsilon}{2\Delta}$ and $A = c_A/(\varepsilon^2 \Delta)$, this is equivalent to:
    \begin{equation*}
        \frac{|M| \cdot |C|^2 \cdot \varepsilon^6}{576 \cdot c_A^2 \Delta} \geq 150 \ln n + 150\Delta + \ln 2.
    \end{equation*}
    Using the fact that $|M| \cdot |C|^2 \geq \varepsilon^{10} \Delta^3$, and recalling $\varepsilon = c_\varepsilon \cdot (\sqrt{\ln n}/\Delta)^{1/16}$ it suffices to argue that:
    \begin{equation*}
        \frac{1}{576 \cdot c^2_A} \cdot c_\varepsilon^{16} \cdot \Delta \sqrt{\ln n} \geq 150 \ln n + 150\Delta + \ln 2,
    \end{equation*}
    which holds by the fact that $\Delta \geq \sqrt{\ln n}$ and the choices of $c_A$ and $c_\varepsilon$.
\end{proof}

\begin{corollary} \label{corollary:matching_lemma}
    Let $M$ be a matching of size $\varepsilon \alpha \Delta$. Then, for any time step $t$:
    \begin{align*}
        &\Pr\left[\sum_{e\in M} Z^{(t)}_e > |M|\right] \leq \frac{1}{2^{150\Delta}n^{150}}, \quad \text{and} \\
        &\Pr\left[\sum_{e\in M} Z^{(t)}_e < (1-c_K \cdot \eps)|M|\text{ and there are no bad endpoints of $M$}\right] \leq \frac{1}{2^{150\Delta}n^{150}}.
    \end{align*}
\end{corollary}
\begin{proof}
    Apply \cref{lemma:matching_lemma} using $C = \Calg$.
\end{proof}

We are now ready to prove \cref{lemma:bounding_Z_per_matching}, which we restate here for convenience:
\BoundingZPerMatching*
\begin{proof}
    Call an edge $e$ \emph{annoying} if no endpoint of $e$ is \emph{bad} at arrival time $t_e$, and $Z^{(t_e-1)}_{e} \notin [1 - c_K \cdot \varepsilon, 1]$. If $Z^{(t_e-1)}_{e} < 1 - c_K \cdot \varepsilon$, call the edge \emph{low-annoying}; if $Z^{(t_e-1)}_{e} > 1$, call the edge \emph{high-annoying}. We have:
    \begin{align*}
        &\Pr \left[ \text{$\exists 2\varepsilon \alpha \Delta$ \emph{annoying} edges in $M$} \right] \leq \\
        &\Pr \left[ \text{$\exists \varepsilon \alpha \Delta$ \emph{low-annoying} edges in $M$}\right] + \Pr \left[ \text{$\exists \varepsilon \alpha \Delta$ \emph{high-annoying} edges in $M$}\right].
    \end{align*}
    We bound the first of these two latter probabilities; the second one can be bounded analogously. Let $\binom{M}{\varepsilon \alpha \Delta}$ be the set of all sub-matchings of size $\varepsilon \alpha \Delta$ of $M$, and note that there are at most $\binom{\alpha \Delta}{\eps\alpha\Delta} \le 2^{\Delta}$ many such submatchings. We have, using \cref{corollary:matching_lemma}:
    \begin{align*}
        \Pr &\left[ \text{$\exists \varepsilon \alpha \Delta$ \emph{low-annoying} edges in $M$}\right] \\
        &= \Pr \left[ \text{$\exists M' \in \binom{M}{\varepsilon\alpha\Delta}$  s.t. all $e \in M'$ are \emph{low-annoying}} \right] \\ 
        &\leq \sum_{M' \in \binom{M}{\varepsilon\alpha\Delta}} \Pr\left[ \text{all $e \in M'$ are \emph{low-annoying}} \right] \\
        &\leq \sum_{M' \in \binom{M}{\varepsilon\alpha\Delta}} \Pr\left[ \sum_{e\in M'} Z^{(t)}_e < (1 - c_K \cdot \varepsilon) \cdot |M| \text{ and there are no bad endpoints of $M$}\right] \\
        &\leq \sum_{M' \in \binom{M}{\varepsilon\alpha\Delta}} \frac{1}{n^{150}}\\
        &
        \leq 2^{\Delta} \cdot \frac{1}{2^{150\Delta}n^{150}} \\
        &\leq \frac{1}{2^{140\Delta}n^{140}}.
    \end{align*}
    Similarly, one can show $\Pr \left[ \text{$\exists \varepsilon \alpha \Delta$ \emph{high-annoying} edges in $M$}\right] \leq 2^{-140\Delta}n^{-140}$ and the desired conclusion follows by summing up the two bounds.
\end{proof}

\subsection{Marked Edges of Bad Vertices: Proof of \cref{lemma:bad_vertex_lemma}} \label{sec:arguing_bad_vertex_lemma}

For convenience, we recall the lemma that is to be proven in this section:
\BadVertexLemma*
We will use the following preliminary result:
\begin{lemma} \label{lemma:helper_for_bad_vertex_lemma}
    Fix a subset $C \subseteq \Calg$ of size $|C| = 2c_K \cdot \varepsilon \Delta$, and a subset $U \subseteq N(v)$ of $|U| = \varepsilon \Delta$ neighbors of a fixed vertex $v$. Let $X_{UC}$ be the event that, at the end of the algorithm, no vertex in $U$ has any free color in $C$ and that no edge was ever \emph{dangerous}. Then: $\Pr[X_{UC}] \leq \frac{1}{2^{100\Delta}n^{100}}$.
\end{lemma}
We first show how assuming this result allows to prove the Bad Vertex Lemma:
\begin{proof}[Proof of \cref{lemma:bad_vertex_lemma}]
First, by \cref{lemma:second_main_lemma}, with high probability in $n$ no vertices ever turn dangerous.
    Fix the vertex $v$, and assume it already turned \emph{bad} at time step $t_0$. Since its \emph{badness} exceeds $2c_K \cdot \varepsilon \Delta$, there are at least this many incident edges that are not colored from $\Calg$, and thus there must exist a subset $C \subseteq \Calg$ of colors of size $|C| = 2c_K \cdot \varepsilon \Delta$, such that none of these colors have been used for $v$ at the end of the algorithm. Let $U_{>t_0} \subseteq N(v)$ be those future neighbors of $v$ that will connect to $v$ at time steps $>t_0$. We observe that, for any neighbor $u \in U_{>t_0}$ connecting to $v$ with edge $e_t = \{u,v\}$, the condition $Z^{(t_{e_t}-1)}_{e_t} = 0$ can hold only if all colors from $C$ have already been used by $u$. Call such neighbors of $v$ \emph{hot}. By the previous discussion, we have that:
    \begin{equation*}
        \Pr\left[ \text{$\exists \varepsilon \Delta$ \emph{hot} neighbors of $v$} \right] \leq \Pr\left[ \text{$\exists U \in \binom{N(v)}{\varepsilon \Delta} \ \exists C \in \binom{\Calg}{2c_K \cdot \varepsilon \Delta}$ s.t. $X_{UC}$} \right],
    \end{equation*}
    where $\binom{N(v)}{\varepsilon \Delta}$ is the set of all subsets of $\varepsilon \Delta$ neighbors of $v$, and $\binom{\Calg}{2c_K \cdot \varepsilon \Delta}$ is the set of all subsets of $\Calg$ of size $2c_K \cdot \varepsilon \Delta$. Continuing by union bound and applying \cref{lemma:helper_for_bad_vertex_lemma}, we have:
    \begin{align*}
        \Pr\left[ \text{$\exists \varepsilon \Delta$ \emph{hot} neighbors of $v$} \right] &\leq \sum_{U \in \binom{N(v)}{\varepsilon \Delta}} \sum_{C \in \binom{\Calg}{2c_K \cdot \varepsilon \Delta}} \Pr[X_{UC}] \\
        &\leq \binom{|N(v)|}{\varepsilon \Delta} \cdot \binom{|\Calg|}{2c_K\eps\Delta} \cdot \frac{1}{2^{100\Delta}n^{100}}
        \\
        &\leq 2^{\Delta} \cdot 2^{\Delta} \cdot \frac{1}{2^{100\Delta}n^{100}}
        \\
        &\leq \frac{1}{n^{100}}.\qedhere{}
    \end{align*}
\end{proof}
We finish the section by proving \cref{lemma:helper_for_bad_vertex_lemma}:
\begin{proof}[Proof of \cref{lemma:helper_for_bad_vertex_lemma}]
    Similarly to the proof of \cref{lemma:second_main_lemma}, we consider the subgraph $G_U \subseteq G$ of the input graph which contains those edges $E_U \subseteq E$ incident to vertices in $U$. Since $G_U$ has maximum degree $\Delta$, its edge set $E_U$ can be partitioned into at most $\Delta + 1$ matchings $M_1,\dots, M_{\Delta + 1}$. Further, each such matching $M_i$ has size at most $\varepsilon \Delta$, as each edge in $G_U$ is incident to a vertex in $U$.

    We call an edge $e \in E_U$ with arrival time $t_e$ \emph{unlucky} if $Z^{(t_e-1)}_{eC} > \frac{|C|}{\Delta}(1-\frac{\varepsilon}{2})$. The first key claim in the proof is the following statement:
    \begin{equation} \label{eq:few_unlucky_per_matching}
        \Pr\left[\text{for all $i$, matching $M_i$ contain at most $\varepsilon^3 \Delta$ \emph{unlucky} edges}\right] \leq \frac{1}{2^{140\Delta}n^{140}}.
    \end{equation}
    To argue this claim, we first note that if $|M_i| < \varepsilon^3 \Delta$ there is nothing to prove. Otherwise, if $|M_i| \geq \varepsilon^3 \Delta$, it is clear for any submatching $M \subseteq M_i$ of size $|M| = \varepsilon^3 \Delta$ that $|M| \cdot |C|^2 \geq \varepsilon^{10} \Delta^3$, such that the Matching Lemma (\cref{lemma:matching_lemma}) is applicable.
    We obtain from the Matching Lemma, for any fixed such submatching $M$ of $M_i$:
    \begin{equation*}
        \Pr\left[ \text{all edges in $M$ are \emph{unlucky}} \right] \leq \frac{1}{2^{150\Delta}n^{150}}.
    \end{equation*}
    To obtain \eqref{eq:few_unlucky_per_matching}, we use the above inequality and union bound over all possible $\leq (\Delta + 1) \cdot \binom{\varepsilon \Delta}{\varepsilon^3 \Delta}$ submatchings $M$ of size $\varepsilon^3 \Delta$ of any matching $M_i$.

    With the above detour completed, we are ready to prove the statement of \cref{lemma:helper_for_bad_vertex_lemma}. Let $X'_{UC}$ be the number of edges incident to vertices in $U$ that are colored using the sub-palette $C$, where we count an edge twice if it connects two vertices from $U$. The statement of the lemma follows if we argue that $X'_{UC} < |C| \cdot |U| = 2c_K \cdot \varepsilon^2 \Delta^2$ (with high probability) at the end of the algorithm, since it implies that at least one vertex from $U$ has not used all colors from $C$ by the end of the algorithm. We proceed by analyzing the evolution of the quantity $X'_{UC}$ over time as the $|E_U| \leq \varepsilon \Delta^2$ edges from $E_U$ arrive. There are different cases to consider, depending on the status of the arriving edge $e_t$ at the time $t_{e_t}$ of its arrival:
    \begin{enumerate}
        \item If $e_t$ is incident to a \emph{bad} vertex on arrival, it might increase $X'_{UC}$ by $|e_t\cap U|\le 2$ if it was assigned a color $c\in C$. We count the number of such edges.
        Since we are only interested in bounding the probability in executions where no vertices turn dangerous (recall \cref{def:bad_vertices_etc}), we know that each vertex will have at most $\alpha \Delta$ bad neighbors.
        Thus at most $\alpha \Delta$ vertices in $U\subseteq N(v)$ are bad,
        having at most $\alpha \Delta^2$ incident edges to them in total. Next, we count the number of edges incident to $U$ such that they are incident to a \emph{bad} vertex which is not in $U$.
        Each vertex in $U$ has at most $\alpha \Delta$ \emph{bad} neighbors, such that there can be at most $\varepsilon \alpha \Delta^2$ such edges. Overall, our count shows that there can be at most $(\alpha + \varepsilon\alpha)\Delta^2$ edges incident to a \emph{bad} vertex arrival, such that the increase in $X'_{UC}$ caused by these edges is also bounded by $3\alpha\Delta^2$.

        \item Otherwise, if $e_t$ connects two \emph{good} vertices on arrival, it might be the case that it is \emph{unlucky}, i.e., $Z^{(t_{e_t}-1)}_{e_tC} > \frac{|C|}{\Delta}(1-\frac{\varepsilon}{2})$. By \eqref{eq:few_unlucky_per_matching}, with good probability there are at most $(\Delta + 1) \cdot \varepsilon^3 \Delta$ such edges. Each of them might increase $X'_{UC}$ by at most $2$, leading to a total increase of at most $2(\Delta + 1) \cdot \varepsilon ^3\Delta \leq 3 \cdot \varepsilon^3 \Delta^2$.

        \item It remains to analyze the case of edges $e_t$ connecting two \emph{good} vertices, and $Z^{(t_{e_t}-1)}_{e_tC} \leq \frac{|C|}{\Delta}(1-\frac{\varepsilon}{2})$. Let $F$ be the set of these edges. For each $f \in F$, define the random indicator $X_f$ which is $1$ iff $f$ is colored using a color from $C$. Since $Z^{(t_{e_t}-1)}_{e_tC} \leq \frac{|C|}{\Delta}(1-\frac{\varepsilon}{2})$ at the time $t_f$ at which $f$ arrives, and $f$ is not incident to any \emph{bad} vertex, we have $\Pr[X_f = 1] \leq \frac{|C|}{\Delta}(1-\frac{\varepsilon}{2}) = 2c_K \cdot \varepsilon(1 - \frac{\varepsilon}{2})$.

        Let $f_1,\dots,f_{|F|}$ be an ordering of the edges in $F$ by their arrival time. Our arguments above have shown that, for any $f_i \in F$, we have $\Pr[X_{f_i} \mid X_{f_1},\dots,X_{f_{i-1}}] \leq 2c_K \cdot \varepsilon(1 - \frac{\varepsilon}{2})$, since the decision w.r.t.\ coloring edge $f$ is independent of previous such decisions.
        Let $s_f = |f\cap U|$ be the amount (either $1$ or $2$) of vertices in $U$ that $f$ touches. We note that $\sum_{f\in F} s_f \le 
        \sum_{u\in U} \deg(u) \le \varepsilon \Delta^2$ and $|F|\le \eps\Delta^2$, so we also have $\E[ \sum_{f\in F} s_f X_f] \leq 2 c_K \cdot \varepsilon^2 \Delta^2 (1 - \frac{\varepsilon}{2})$. Thus, by standard coupling arguments, using Chernoff-Hoeffding bounds (noting that $s_f X_f \in [0,2]$) we obtain that:
        \begin{align*}
            \Pr \left[ \sum_{f \in F} s_f X_f \geq (2 c_K + 2\varepsilon) \cdot \varepsilon^2 \Delta^2 \left (1 - \frac{\varepsilon}{2} \right) \right] &= \Pr \left[ \sum_{f \in F} s_f X_f \geq 2 c_K \cdot \varepsilon^2 \Delta^2 \left( 1 - \frac{\varepsilon}{2} \right) \cdot \left (1 + \frac{\varepsilon}{c_K} \right) \right] \\
            &\leq \exp \left( - \frac{2 \cdot \varepsilon^4 \Delta^2 \left( 1 - \frac{\varepsilon}{2} \right)}{12 \cdot c_K} \right) \\
            &\leq \frac{1}{2^{110\Delta}n^{110}}.
        \end{align*}
        Hence, with probability at least $1 - 2^{-110\Delta}n^{-110}$, the total increase of $X'_{UC}$ caused by edges considered in this case is at most $(2 c_K + 2\varepsilon) \cdot \varepsilon^2 \Delta^2 \left (1 - \frac{\varepsilon}{2} \right)$.
    \end{enumerate}
    Recall that our objective is to prove that, at the end of the algorithm, the quantity $X'_{UC}$ is smaller than $|C| \cdot |U| = 2c_K \cdot \varepsilon^2 \Delta^2$. By the above discussion, the total increase is given (with probability at least $1-2^{-100\Delta}n^{-100}$, union bounding over all the above events) by:
    \begin{align*}
        &3\alpha\Delta^2 + 3 \cdot \varepsilon^3 \Delta^2 + (2 c_K + 2\varepsilon) \cdot \varepsilon^2 \Delta^2 \left (1 - \frac{\varepsilon}{2} \right) \\
        =\ &2c_K \cdot \varepsilon^2 \Delta^2 - ((c_K+\eps) \cdot \varepsilon^3 \Delta^2 - 
        3\alpha \Delta^2 - 3\eps^3 \Delta^2)
        \\
        <\ &
        2c_K \cdot \varepsilon^2 \Delta^2.
    \end{align*}
    The last inequality follows by the choice of $c_K \ge 100$ and $\alpha \le \eps^{3}/100$. 
\end{proof}

\appendix
\section*{APPENDIX}
\section{Lower Bounds}\label{sec:randomized-greedy}

\subsection{Online List Edge Coloring is Harder than Online Edge Coloring}
\label{sec:list-edge-coloring-lower-bound}

In the \emph{list edge coloring} problem, edges of the graph are revealed with an associated palette. 
An online algorithm must assign each edge a color from its designated palette (while guaranteeing that neighboring edges receive different colors).
In static settings it is known by the celebrated work of Kahn \cite{kahn1996asymptotically} that palettes of size $(1+\eps)\Delta$ suffice to color the graph (with $\eps\to 0$ as $\Delta\to \infty$).
Abusing notation slightly, we refer to this as a $\Delta(1+\eps)$ list edge coloring.
In online settings, \cite{blikstad2024online} showed that the bounds conjectured by \cite{bar1992greedy}, namely, a $(1+o(1))\Delta$-edge-coloring for graphs of maximum degree $\Delta=\omega(\log n)$ attained randomly, also holds for the more general list edge coloring problem. This may lead one to suspect that these problems are in some sense equivalent.

Here we provide a separation between these two problems. 
We prove that the tight bounds achievable for online edge coloring (\Cref{result:adaptive,result:oblivious}) are not achievable for the list edge coloring problem. Indeed, we observe that for higher degree thresholds, even palettes of size $2\Delta-1$, which suffice for a greedy algorithm, are needed.

The basic instance is again two stars of degree $\Delta-1$, where the roots $u,v$ are finally connected by an edge $\{u,v\}$, arriving last.
Each edge of the stars has a distinct palettes (of size at most $2\Delta-1$), resulting in these having distinct colrs.
We show that if the last edge $e:=\{u,v\}$ has palette $|P_e|$ size strictly less than $2\Delta-1$, then it cannot be colored (at least probabilistically).

For deterministic lower bounds: 
The final edge $e=\{u,v\}$ arrives with palette $P_e$ a subset of $|P_e|\leq 2\Delta-2$ of the neighboring edges' previously-assigned colors, and so it cannot colored.

\begin{theorem}
    No deterministic online list edge coloring algorithm succeeds with edge palettes of size strictly less than $2\Delta-1$, even for graphs with maximum degree $\Delta=\Omega(n)$.
\end{theorem}

For randomized lower bounds: 
The palette $P_e$ of  the last edge $e=\{u,v\}$ is obtained by picking a color uniformly at random from the palettes of each of some $|P_e|\leq 2\Delta-2$ neighboring edges' palettes. 
Since each such palette has size at most $2\Delta-1$, the edge $\{u,v\}$ cannot be colored from its palette precisely when the random palette of $\{u,v\}$ contains only colors used by neighboring edges, which happens with probability $(1/c\Delta)^{2\Delta-2}=\Delta^{-\Theta(\Delta)}$.
Thus, if we repeat this construction some $\Delta^{\Theta(\Delta)}$ many times, any online algorithm has a constant probability of failing. The above requires $\Delta^{\Theta(\Delta)}$ many nodes, which is achievable already for $\Delta=O\big(\frac{\log n}{\log \log n}\big)$.
Again, this makes a list edge coloring counterpart to \Cref{result:oblivious} impossible, and shows that the list edge coloring result of \cite{blikstad2024online} is almost tight.

\begin{theorem}
    No randomized online list edge coloring algorithm  succeeds (against oblivious adversaries) with edge palettes of size strictly less than $2\Delta-1$, even for graphs with maximum degree $\Delta=\Omega\big(\frac{\log n}{\log \log n}\big)$.
\end{theorem}

\subsection{Randomized Greedy Fails When $\Delta = O(\log n)$}
\label{sec:oblivious_lowerbound}
In this section we prove that the candidate \emph{randomized greedy} algorithm of \cite{bar1992greedy} does not yield a $(1+o(1))\Delta$-edge-coloring matching the bounds of our \Cref{result:oblivious}. 
We recall that randomized greedy (see \cref{alg:randomized-greedy-edge-coloring}) picks a uniformly random color from the palettes of colors available to both its endpoints, with all palettes initialized to some common palette $\Calg$.

\begin{algorithm}[ht!]
	\caption{Randomized Greedy Edge Coloring Algorithm \cite{bar1992greedy}}
	\label{alg:randomized-greedy-edge-coloring}
	\begin{algorithmic}[1]
        \Statex \underline{\smash{\textbf{Input:}}} Vertex set $V$ and maximum degree $\Delta\in \mathbb{Z}_{\geq 0}$ of the graph to arrive
        \Statex 
        \textbf{\underline{Initialization:}} For each vertex $v\in V$: Set a palette of colors
         $P_v\gets \C_{\mathrm{alg}}$.
        \Statex 
		\For{\textbf{each} online edge $e_t=\{u,v\}$ on arrival} 
            \State Assign $e_t$ a uniformly random color $c$ from $P_u\cap P_v$ (\textbf{if none exist: fail})
            \State $P_u\gets P_u \setminus \{c\}$ and $P_v\gets P_v\setminus \{c\}$.
		\EndFor
	\end{algorithmic}
\end{algorithm}

We observe that this algorithm cannot even beat the (non-randomized) greedy algorithm's $2\Delta-1$ colors for sufficiently small $\Delta=O(\log n)$, on obliviously generated graphs.

\begin{thm}\label{thm:rg-poor-obliviously}
For any $n\geq 16$ and $2\leq \Delta\leq \frac{1}{4}\log_2 n+1$, randomized greedy (\Cref{alg:randomized-greedy-edge-coloring}) fails with any palette of size $2\Delta-2$ or smaller on some $n$-node graph of maximum degree $\Delta$ generated in advance, by an oblivious adversary.
\end{thm}
\begin{proof}
    Consider a run of randomized greedy with a palette $\Calg$ of size $m \leq 2\Delta-2$, on a (sub)graph consisting of two stars with $\Delta-1$ leaves, rooted at $u$ and $v$, with an edge $e_t=\{u,v\}$ arriving after the stars' edges.\footnote{\cite{bar1992greedy} used this construction, but only proved with it that randomized greedy needs $|\Calg|\geq \Delta+\Omega(\sqrt{\Delta})$ colors.}
    Fix the remaining palette $P_u$ of $u$ after coloring both stars. 
    Note that if $P_u\cap P_v = \emptyset$ before edge $e_t$ arrives, then randomized greedy fails to color $e_t$.
    So, the algorithm fails, for example, if the first $k:=|\Calg\setminus P_u|$ edges of $v$ are assigned colors in $\Calg\setminus P_u$. (Note that $k = |\Calg|-(\Delta-1) = m-(\Delta-1) \leq (2\Delta-2)-(\Delta-1)=\Delta-1$, so $v$ has at least $k$ edges in its star.) Consequently:
    \begin{align*}
        \Pr[\textrm{randomized greedy fails on $\{u,v\}$}] & \geq \Pr[\textrm{first $k$ edges of $v$ colored using $\Calg\setminus P_u$}] \\
        & =\prod_{i=1}^{k}\frac{k+1-i}{m+1-i} \\
        & = \frac{k!}{m\cdot(m-1)\cdot \dots \cdot(m-k+1)} \\
        & = \frac{k!}{m!/(m-k)!} \\
        & = \frac{1}{\binom{m}{k}} \\
        & \geq 1/2^m.
    \end{align*}
    Therefore, for $\Delta<\frac{1}{4}\log_2 n+1$, we have that $m\leq 2\Delta-2\leq \frac{1}{2}\log_2 n$, and so the probability of failure on one such star is at least $1/2^m\geq 1/\sqrt{n}$.
    Now, consider a graph containing $\lfloor \sqrt{n}/2 \rfloor$ copies of this graph, consisting of $\lfloor \sqrt{n}/2\rfloor(2\Delta)\leq \sqrt{n}\Delta \leq n$ nodes (using that $\Delta\leq \frac{1}{4}\log_2n+1\leq \sqrt{n}$ for all $n\geq 1$), and a further $n-\lfloor \sqrt{n}/2\rfloor(2\Delta)$ isolated nodes.
    Then, by independence of choices of randomized greedy, the above is an $n$-node graph with $n\geq 16$ 
    of maximum degree $\Delta=\frac{1}{4}\log_2n+1$, on which randomized greedy fails with constant probability, at least $1-(1-1/\sqrt{n})^{\lfloor \sqrt{n}/2\rfloor}\approx 1/e^2$.
\end{proof}

\begin{remark}
    The above bad example extends to random-order edges arrivals (where in each gadget the probability of failure decreases by a factor of $1/(2\Delta-1)$, which is precisely the probability of $\{u,v\}$ arriving last among the $2\Delta-1$ edges.
    Repeating the above construction some $\Theta(\Delta2^m)$ times results in a constant failure probability on $n$-node graphs with maximum degree $\Delta=O(\log n)$ under random-order arrivals. This shows that the analysis of \cite{dudeja2025randomized} for \Cref{alg:randomized-greedy-edge-coloring} under random-order arrivals is asymptotically optimal.
\end{remark}

\begin{remark}[Emergency palettes]
This lower bound example demonstrates that algorithms treating colors symmetrically and coloring disjoint stars independently will inevitably fail. This reveals the necessity for asymmetric treatment of the colors in our approach. We achieve this asymmetry in \cref{alg:edge-coloring,alg:edge-coloring-oblivious} by implementing a backup/emergency color palette $\Cgreedy$
that reserves certain colors for use only when necessary.
\end{remark}

\subsection{Randomized Greedy is Unlikely To Work Against Adaptive Adversaries}
\label{sec:experiments}
In this section we explain why we believe that randomized greedy (\cref{alg:randomized-greedy-edge-coloring}) is problematic if the adversary is adaptive, which motivates our use of a different algorithm to achieve our \cref{result:adaptive}. In particular, we have implemented experiments that run randomized greedy on a specific adaptively-generated graph that seem to suggest that randomized greedy fails if the palette is smaller than $\frac{e}{e-1}\Delta$.

Our insight is that the randomized greedy algorithm seemingly amplifies 
``bias'' (preferring some colors over others) on specific adversarially adaptively-generated graphs, despite picking colors uniformly at random from the common palettes $P_u \cap P_v$ for arriving edge $\{u,v\}$.  This is in contrast to our \cref{alg:edge-coloring,alg:edge-coloring-oblivious}, where a color is picked from the common palette $P_{u}\cap P_{v}$ not uniformly, but instead with weighted probabilities based on the current execution. These weighted (or biased) probabilities work in our favor, counteracting any bias the adaptive adversary could try to introduce.

\begin{remark}
\label{remark:simulation}
What we write below is \textbf{not} a formal lower bound ruling our randomized greedy. Instead, our aim is to provide some intuition explaining why we believe that randomized greedy is not the right ``deterministic'' (equivalently, randomized against adaptive adversaries) algorithm to look at. This also motivates our use of weighted probabilities in \cref{alg:edge-coloring,alg:edge-coloring-oblivious}.
While we do not formally prove that randomized greedy does not work on the hard instance described below, we have implemented a simulation of randomized greedy on this hard instance. In our simulations (when $\Delta \approx 10^{4}$ and $n\approx 10^{16}$) randomized greedy seemingly fails whenever the palette has size $\le 1.58\Delta$ (but succeeds for palettes of size $\ge 1.6 \Delta$), supporting our intuition that randomized greedy needs at least $\approx (\frac{e}{e-1}) \Delta\approx 1.582\Delta$ colors to work against adaptive adversaries.
\end{remark}

\paragraph{Hard (Adaptive) Instance For Randomized Greedy.}
We describe a construction with $\Delta \gg \log n$ where we believe randomized greedy requires at least $\approx \frac{e}{e-1}\Delta$ colors. The adversary first generates many $\Delta-1$ degree stars---randomized greedy will color each star independently and uniformly at random from a palette $\Calg$. Split up the palette into two equally sized parts $L\, \dot{\cup}\, R = \Calg$. For a vertex $v$ the center of a star, let $\ell_{v} = \frac{|L\cap P_{v}|}{|P_v|}$ and $r_{v} = \frac{|R\cap P_{v}|}{|P_v|}$ be the fraction of available colors in $L$ and $R$ respectively. Note that $\ell_v +r_v = 1$, and define the \emph{bias} of $v$ to be $b_v = \max\{\ell_{v},r_v\} - \frac{1}{2}$. For most stars, we expect $b_v = \Theta(\frac{1}{\sqrt{\Delta}})$ (i.e., $\ell_v$ is around one standard deviation away from its mean).

The construction is now easy: discard all vertices where $\ell_{v} \le r_v$, and, in a bottom up fashion, build a balanced rooted tree with the remaining stars as leafs. That is, the remaining stars are in layer $0$, and layer $i+1$ is constructed by creating new vertices $v$ that are each connected to $\Delta-1$ vertices from layer $i$. Note that the only \emph{adaptivity} is choosing what leafs to keep (discarding those with $\ell_v < r_v$), building the tree afterwards can be done by an oblivious adversary.  We will only need $O(\log \Delta)$ layers, so in total $\Delta^{O(\log \Delta)}$ nodes (for example setting $\Delta \approx 2^{\sqrt{\log n}}\gg \log^{100} n$ would work).

\paragraph{Bias Amplification for Randomized Greedy.}
Here we will sketch how using randomized greedy intuitively leads to the bias growing exponentially between layers (alternating between preferring colors in $L$ versus $R$). If this was the case, after $O(\log \Delta)$ layers vertices would have bias $b_v = 1$, that is either all of $L$ or all of $R$ is unavailable. Then a vertex on the next layer will not be able to successfully color all its neighbors as $|L| = |R| = |\Calg|/2 \ll \Delta$.

Consider a vertex $v$ connected to $\Delta-1$ nodes $u_1, \ldots, u_{\Delta-1}$ all with $\ell_{u_i} = \frac{1}{2} + \eps$ and $r_{u_i} = \frac{1}{2} - \eps$ (so their bias is $b_{u_i} = \eps$). We will sketch an argument that the bias of $v$ is, intuitively, twice as large, i.e.~$\ge 2\eps$. When the edge $\{v,u_1\}$ arrives, it will use a color from $L$ with probability $\frac{1}{2}+\eps$ or a color from $R$ with probability $\frac{1}{2}-\eps$. Let us for now (incorrectly) imagine this happens for the other edges $\{v,u_{i}\}$ similarly and independently. Then we would expect $v$ to have $(\frac{1}{2}+\eps)(\Delta-1)$ incident edges colored from $L$. Note that $|L|= |\Calg|/2 = \frac{e}{2(e-1)}\Delta$ and $v$ has now $|P_v| = |\Calg|-(\Delta-1) \approx \frac{1}{e-1}\Delta$ remaining free colors. We get:
\begin{equation*}
\ell_{v} = \frac{|L\cap P_e|}{|P_e|} \approx
\frac{\frac{e}{2(e-1)}\Delta-\frac{1}{2}\Delta-\eps\Delta}{\frac{1}{e-1}\Delta}
= \frac{1}{2}-(e-1)\eps.
\end{equation*}
That is, with our rough estimate we would expect the bias of $v$ to be $b_v = (e-1)\eps > 1.5\eps$ if it connects to $\Delta-1$ vertices with bias $b_{u_i} = \eps$.

The above analysis makes an unrealistic assumption that the colors of edges $\{v,u_i\}$ are independent. Indeed this is not the case, if among $u_1, \ldots, u_{k}$ many of the edges where colored using $L$,  then the edge $\{v,u_{k+1}\}$ is less likely to be colored from $L$.
The distribution of how many incident edges (at $v$) are colored from $L$ is more complicated to reason about. 
While not quite the same, this distribution intuitively behaves similarly to Wallenius' noncentral hypergeometric distribution\footnote{This distribution represents a biased urn model where each ball has a weight. With $m_1$ red balls (weight $\omega_1$) and $m_2$ white balls (weight $\omega_2$), we draw $n$ balls sequentially. At each draw, the probability of selecting any specific ball equals its weight divided by the total weight of all remaining balls. } with weights $\omega_1 = \frac{1}{2}+\eps$, $\omega_2 = \frac{1}{2}-\eps$, and $m_1 = |L|$, $m_2 = |R|$ and $n = \Delta-1$. Under the (still not quite correct, but significantly more realistic) assumption that the colors are drawn from Wallenius' noncentral hypergeometric distribution, the bias would still grow exponentially as above, \emph{but only if the palette is smaller than $\frac{e}{e-1}\Delta$}.\footnote{At \url{https://www.desmos.com/calculator/rmrf5jaevv} we show an analytic derivation of this threshold.} This intuition (while not formal) agrees with our simulations (see also \cref{remark:simulation}).

An alternative way to reason about this process is to set it up as a stochastic differential equation (SDE), which generally by standard arguments such as Kurtz's theorem \cite{kurtz1970solutions}, we expect to behave like the ordinary (deterministic) differential equation (ODE) obtained by replacing the discrete (expected) derivative by the its continuous (deterministic) counterpart, and that the SDE is with high probability at all times near the ODE's solution. Unfortunately, the obtained ODE seems challenging to solve, and we can only verify numerically that this ODE converges to a scenario increasing bias when $|\Calg|\leq (e/(e-1))\Delta$.
Unfortunately, this ODE seems challenging to solve other than numerically, and so we omit the details.

\begin{remark}
One could ask if a backup palette could help randomized greedy on the above instance? We believe the answer is no---that it would not help---as the root of the tree would get too many edges forwarded to the backup palette (also supported by implementating and simulating this).
\end{remark}

\paragraph{Bias Cancellation in \cref{alg:edge-coloring,alg:edge-coloring-oblivious}.}
Our \cref{alg:edge-coloring,alg:edge-coloring-oblivious} would not run into the same issues as randomized greedy on this tricky instance. If many edges incident to some vertex $v$ are all more biased toward picking colors from $L$, the algorithm would naturally scale up the probabilities $P_{ec}$ for picking the remaining colors $c\in L$ \emph{more} than that for colors $c\in R$. Even though vertex $v$ will eventually have fewer colors free in $L$ than in $R$, the colors $c\in L$ have larger weights $P_{ec}$ assigned to them (since they where scaled up more), and thus a future arriving edge $e$ will still be roughly equally likely to be colored from $L$ or $R$. That is, the vertex~$v$ is (from an another vertex's perspective) less, and not more, biased---preventing any bias amplification.

\bibliographystyle{alpha}
\bibliography{abb,ultimate,bibliography}

\end{document}